\documentclass{article}
\pdfoutput=1

\usepackage[colorlinks=true,breaklinks=true,bookmarks=true,urlcolor=blue,citecolor=blue,linkcolor=blue,bookmarksopen=false,draft=false]{hyperref}
\usepackage{appendix}
\usepackage{color}
\usepackage{latexsym}
\usepackage{dsfont}
\usepackage{float}
\usepackage{algorithm,algpseudocode}
\usepackage{graphicx}
\usepackage{amsmath, amsfonts,amssymb,theorem,euscript,array,enumerate,amsfonts,mathrsfs}
\usepackage{subcaption}
\usepackage{ulem}
\usepackage{bbm}
\usepackage[T1]{fontenc}
\usepackage[latin1]{inputenc}
\usepackage{hhline}
\usepackage{verbatim}
\usepackage{eurosym}
\usepackage{booktabs, multirow}       
\usepackage{lineno}
\usepackage{enumitem}

\usepackage{endnotes}
\let\footnote=\endnote

\usepackage{natbib}
\bibpunct[, ]{[}{]}{,}{n}{}{,}%

\numberwithin{equation}{section}

\newtheorem{Theorem}{Theorem}[section]
\newtheorem{Definition}[Theorem]{Definition}
\newtheorem{Proposition}[Theorem]{Proposition}
\newtheorem{Assumption}[Theorem]{Assumption}
\newtheorem{Lemma}[Theorem]{Lemma}

\newtheorem{Remark}[Theorem]{Remark}
\newtheorem{Example}[Theorem]{Example}

\newenvironment{proof}[1][{\it Proof.}]{\begin{trivlist}
\item[\hskip \labelsep {\bfseries #1}]}{ \hfill
$\Box$\end{trivlist}\vskip -0.2 cm}

\theorembodyfont{\upshape}

\newcommand{\R}{\mathbb{R}}

\newcommand{\E}{\mathcal{E}}

\def\esssup_#1{\underset{#1}{\mathrm{ess\,sup\, }}}
\def\essinf_#1{\underset{#1}{\mathrm{ess\,inf\, }}}
\def\argmax_#1{\underset{#1}{\mathrm{arg\,max\, }}}
\def\argmin_#1{\underset{#1}{\mathrm{arg\,min\, }}}

\def \reff#1{{\rm(\ref{#1})}}

\def\b1{\bf 1}

\def \A{\mathbb{A}}

\def \N{\mathbb{N}}
\def \R{\mathbb{R}}

\def \E{\mathbb{E}}

\def \A{{\cal A}}
\def \Ac{{\cal A}}
\def \Bc{{\cal B}}
\def \Cc{{\cal C}}

\def \Gc{{\cal G}}

\def \Ic{{\cal I}}

\def \Lc{{\cal L}}
\def \Pc{{\cal P}}
\def \Mc{{\cal M}}

\def \Sc{{\cal S}}

\def \Vc{{\cal V}}
\def \Wc{{\cal W}}

\def \Vc{{\cal V}}

\def \reff#1{{\rm(\ref{#1})}}
\def \beqs{\begin{eqnarray*}}
\def \enqs{\end{eqnarray*}}
\def \beq{\begin{eqnarray}}
\def \enq{\end{eqnarray}}

\begin{document}


\title{Graphon Mean-Field Control for Cooperative Multi-Agent Reinforcement Learning}

\author{Yuanquan Hu \thanks{huyq21@mails.tsinghua.edu.cn} \and Xiaoli Wei \thanks{xiaoli\_wei@sz.tsinghua.edu.cn} \and Junji Yan \thanks{yan-jj21@mails.tsinghua.edu.cn} \and Hengxi Zhang \thanks{zhanghx20@mails.tsinghua.edu.cn} }

\date{}






\maketitle 




\begin{abstract}

  The marriage between mean-field theory and reinforcement learning has shown a great capacity to solve large-scale control problems with homogeneous agents. To break the homogeneity restriction of mean-field theory, a recent interest is to introduce graphon  theory to the mean-field paradigm. In this paper, we propose a graphon mean-field control (GMFC) framework to approximate cooperative multi-agent reinforcement learning (MARL) with nonuniform interactions and show that the approximate order is of $\mathcal{O}(\frac{1}{\sqrt{N}})$, with $N$ the number of agents. 
  By discretizing the graphon index of GMFC, we further introduce a smaller class of GMFC called block GMFC, which is shown to well approximate cooperative MARL. 
  Our empirical studies on several examples demonstrate that our GMFC approach is comparable with the state-of-art MARL algorithms while enjoying better scalability. The code of our experiments is available at https://github.com/huyq/GMFC/.
\end{abstract}








\section{Introduction}

Multi-agent reinforcement learning (MARL) has found various applications in the field of transportation and simulating \cite{w2016multi, adler2002cooperative}, stock price analyzing and trading \cite{lee2002stock, lee2007multiagent}, wireless communication  networks \cite{1284411, choi2009distributed, cui2019multi}, and learning behaviors in social dilemmas \cite{leibo2017multi, hughes2018inequity, lerer2017maintaining}. MARL, however, becomes intractable due to the complex interactions among agents as the number of agents increases.

A recent tractable approach is a mean-field approach by considering MARL in the regime with a large number of homogeneous agents under weak interactions \cite{flyvbjerg1993mean}. According to the number of agents and  learning goals, there are three subtle types of mean-field theories for MARL. The first one is called mean-field MARL (MF-MARL), which refers to the empirical average of the states or actions of a {\it finite} population. For example, \cite{YLLZZ2018} proposes to approximate interactions within the population of agents by averaging the actions of the overall population or neighboring agents. \cite{Lietal2021} proposes a mean-field proximal policy optimization algorithm for a class of MARL with permutation invariance. The second one is called mean-field game (MFG), which describes the asymptotic limit of non-cooperative stochastic games as the number of agents goes to infinity \cite{lasry2007mean, HMC2006, carmona2013probabilistic}. 
 Recently, a rapidly growing literature studies MFG for noncooperative MARL either in a model-based way \cite{YMMS2010, CL2018, HS2019} or by a model-free approach \cite{GHXZ2019, SM2019, EPLGP2020, CK2021mfg, PBK2021}. The third one is called mean-field control (MFC), which is closely related to MFG yet different from MFG in terms of learning goals.
For cooperative MFC, the Bellman equation for the value function is defined on an enlarged space of probability measures, and MFC is always reformulated as a new Markov decision process (MDP) with continuous state-action space \cite{motte2022mean}. \cite{CLT2019} shows the existence of optimal policies for MFC in the form of mean-field MDP and adapts classical reinforcement learning (RL) methods to the mean-field setups.
\cite{GGWX2020} approximates MARL by a MFC approach, and proposes a model-free kernel-based Q-learning algorithm (MFC-K-Q) that enjoys a linear convergence rate and is independent of the number of agents.
\cite{PBK2021} presents a model-based RL algorithm M3-UCRL for MFC with a general regret bound. 
\cite{AFL2022} proposes a unified two-timescale learning framework for MFG and MFC by tuning the ratio of learning rates of $Q$ function and the population state distribution.

One restriction of the mean-field theory is that it eliminates the difference among agents and interactions between agents are assumed to be uniform. However, in many real world scenarios, strategic interactions between agents are not always uniform and rely on the relative positions of agents. To develop scalable learning algorithms for multi-agent systems with heterogeneous agents, one approach is to exploit the local network structure of agents \cite{QWL2020,LQHW2021}.
Another approach is to consider mean-field systems on large graphs and their asymptotic limits, which leads to graphon mean-field theory \cite{L2012}. 
So far, most existing works on graphon mean-field theory consider either diffusion processes without learning in continuous time or non-cooperative graphon mean-field game (GMFG) in discrete time. \cite{BCW2020} considers uncontrolled graphon mean-field systems in continuous time. ~\cite{Delarue2017} studies MFG on an  Erd\"os-R\'enyi graph. \cite{EW2014} studies the convergence of weighted empirical measures described by stochastic differential equations. 
\cite{GCN2020} studies propagation of chaos of weakly interacting particles on general graph sequences. \cite{CH2019} considers general GMFG and studies $\varepsilon$-Nash equilibria of the multi-agent system by a PDE approach in continuous time. \cite{LS2022} studies stochastic games on large graphs and their graphon limits. It shows that GMFG is viewed as a special case of MFG by viewing the label of agents as a component of the state process. \cite{GC2019a, GC2019b} study continuous-time cooperative graphon mean-field systems with linear dynamics.
On the other hand, \cite{CCGL2022} studies static finite-agent network games and their associated graphon games. \cite{VMV2021} provides a sequential decomposition algorithm to find Nash equilibria of discrete-time GMFG. \cite{CK2021} constructs a discrete-time learning GMFG framework to analyze approximate Nash equilibria for MARL with nonuniform interactions. However, little is focused on learning cooperative graphon mean-field systems in discrete time, except for \cite{MAAU2021, MAU2022} on particular forms of nonuniform interactions among agents. \cite{MAU2022} proves that when the reward is affine in the state distribution and action distribution, MARL with nonuniform interactions can still be approximated by classic MFC. 
\cite{MAAU2021} considers multi-class MARL, where agents belonging to the same class are homogeneous. 
In contrast, we consider a general discrete-time GMFC framework under which agents are allowed to interact non-uniformly on any network captured by a graphon.

\paragraph{Our Work}
In this work, we propose a general discrete-time GMFC framework to approximate cooperative MARL on large graphs by combining classic MFC and network games. 
Theoretically, we first show that GMFC can be reformulated as a new MDP with deterministic dynamics and infinite-dimensional state-action space, hence the Bellman equation for Q function is established on the space of probability measure ensembles. It shows that GMFC approximates cooperative MARL well in terms of both value function and optimal policies. The approximation error is at order $\mathcal{O}(1/\sqrt{N})$, where $N$ is the number of agents. Furthermore, instead of learning infinite-dimensional GMFC directly, we introduce a smaller class called block GMFC by discretizing the graphon index, which can be recast as a new MDP with deterministic dynamic and finite-dimensional continuous state-action space. We show that the optimal policy ensemble learned from block GMFC is near optimal for cooperative MARL. To deploy the policy ensemble in the finite-agent system, we directly sample from the action distribution in the blocks. Empirically, our experiments in Section \ref{experiment} demonstrate that when the number of agents becomes large, the mean episode reward of MARL becomes increasingly close to that of block GMFC, which verifies our theoretical findings. Furthermore, our block GMFC approach achieves comparable performances with other popular existing MARL algorithms in the finite-agent setting.

\paragraph{Outline} The rest of the paper is organized as follows. Section \ref{sec:MFMARL_dense_graph} recalls basic notations of graphons and introduces the setup of cooperative MARL with nonuniform interactions and its asymptotic limit called GMFC. Section \ref{sec:theoretical} connects cooperative MARL and GMFC, introduces block GMFC for efficient algorithm design, and builds its connection with cooperative MARL. The main theoretical proofs are presented in Section \ref{sec:proofs}. Section \ref{experiment} tests the performance of block GMFC experimentally. 

\section{Mean-Field MARL on Dense Graphs} \label{sec:MFMARL_dense_graph}

\subsection{Preliminary: Graphon Theory}

In the following, we consider a cooperative multi-agent system and its associated mean-field limit.
In this system, each agent is affected by all others, with different agents exerting different effects on her.
This multi-agent system with  $N$ agents can be described by a weighted graph $G_N=(\Vc_N, \mathcal{E}_N)$, where the vertex set $\Vc_N=\{1, \ldots, N\}$ and the edge set $\mathcal{E}_N$ represent agents and the interactions between agents, respectively. 
To study the limit of the multi-agent system as $N$ goes to infinity, we adopt the graphon theory introduced in \cite{L2012} used to characterize the limit behavior of dense graph sequences. Therefore, throughout the paper, we assume the graph $G_N$ is dense and leave sparse graphs for future study.

In general, a graphon is represented by a bounded symmetric measurable function $W:$ $\mathcal{I} \times \mathcal{I} \to \mathcal{I}$, with $\mathcal{I} = [0,1]$. 
We denote by $\Wc$ the space of all graphons and equip the space $\Wc$ with the cut norm $\|\cdot\|_{\square}$
\beqs
\|W\|_{\square} =\sup_{S, T \subset \mathcal{I}} \biggl|\int_{S \times T} W(\alpha, \beta) d\alpha d \beta\biggl|.
\enqs
It is worth noting that each weighted graph $G_N=(\Vc_N, \mathcal{E}_N)$ is uniquely determined by a step-graphon $W_N$
\beqs
W_N(\alpha, \beta) = W_N\Big(\frac{\lceil N\alpha \rceil}{N}, \frac{\lceil N\beta \rceil}{N}\Big).
\enqs
We assume that the sequence of $W_N$ converges to a graphon $W$ in cut norm as the number of agents $N$ goes to infinity, which is crucial for the convergence analysis of cooperative MARL in Section \ref{sec:theoretical}.
\begin{Assumption}\label{assm:WN}
The sequence $(W_N)_{N \in \N}$ converges in cut norm to some graphon $W \in \Wc$ such that
\beqs
\|W_N - W\|_{\square} \to 0.
\enqs
\end{Assumption}

Some common examples of graphons include 
\begin{enumerate}[label={{\upshape \arabic*)}}]
    \item Erd\H{o}s R\'enyi:  $W(\alpha, \beta) = p$, $0 \leq p \leq 1$, $\alpha, \beta \in \mathcal{I}$;
    \item Stochastic block model: 
    \beqs 
        W(\alpha, \beta) = \left\{\begin{array}{rll} p  & \mbox{if}\; 0 \leqslant \alpha, \beta \leqslant 0.5 \;\mbox{or}\; 0.5 \leqslant \alpha, \beta \leqslant 1,\\
        q  & \mbox{otherwise},
    \end{array}
    \right.
    \enqs 
    where $p$ represents the intra-community interaction and $q$ the inter-community interaction;
    \item Random geometric graphon: $W(\alpha, \beta) = f(\min(|\beta - \alpha|, 1 - |\beta - \alpha|))$, where $f: [0, 0.5] \to [0, 1]$ is a non-increasing function.
\end{enumerate}

\subsection{Cooperative MARL with Nonuniform Interactions}\label{subsec:local_setting}
In this section, we facilitate the analysis of MARL by considering a particular class of MARL with nonuniform interactions, where each agent interacts with all other agents via the aggregated weighted mean-field effect of the population of all agents. 

 Recall that we use the weighted graph $G_N=(\Vc_N, \mathcal{E}_N)$ to represent the multi-agent system, in which agents are cooperative and coordinated by a central controller. They share a finite state space $\Sc$ and take actions from a finite action space $\mathcal{A}$. We denote by $\Pc(\Sc)$ and $\Pc(\Ac)$ the space of all probability measures on $\Sc$ and $\Ac$, respectively. Furthermore, denote by $\Bc(\Sc)$  the space of all Borel measures on $\Sc$.

For each agent $i$, the {\it neighborhood empirical measure} is given by
\beq \label{equ:densegraph}
    \mu^{i,W_N}_t(\cdot): = {\frac{1}{N}} \sum_{j \in \Vc_N} \xi_{i, j}^N  \delta_{s_t^j}(\cdot),
    \enq
    where $\delta_{s_t^j}$ denotes Dirac measure at $s_t^j$, and $\xi_{i, j}^N$ describing the interaction between agents $i$ and $j$ is taken as either
    \begin{align} \label{C1}
    \xi_{ij}^N = W_N(\frac{i}{N}, \frac{j}{N}) \tag{C1}
    \end{align} or
    \begin{align} \label{C2}
    \xi_{i,j}^N  \sim \mbox{Bernoulli}(W_N(\frac{i}{N}, \frac{j}{N})). \tag{C2}
    \end{align} 

At each step $t=0,1, \cdots,$ if agent $i$, $i \in [N]$ at state $s^i_t \in \Sc$  takes an action $a^i_t \in \mathcal{A}$, then she will receive a reward
\vspace{-0.2cm}
\begin{eqnarray}\label{eq:N_agent_reward_heter}
r\Big(s^i_t, \,\,\mu_t^{i, W_N}, \,\,a^i_t\Big), \quad i \in [N],
\end{eqnarray}
where $r: \Sc \times \Bc(\Sc) \times \Ac \to \R$,
and she will change to a new state $s^i_{t + 1}$ according to a transition probability such that
\begin{eqnarray}\label{eq:N_agent_transition_heter}
s^i_{t + 1}\sim {P}\left.\Big(\cdot\,\right\vert\,s^i_t, \, \mu_t^{i, W_N}, \,\,a^i_t\Big), \quad i \in [N], \; s^i_0 \sim \mu \in \Pc(\Sc),
\end{eqnarray}
where $P: \Sc \times \Bc(\Sc) \times \Ac \to \Pc(\Sc)$.

\eqref{eq:N_agent_reward_heter}-\eqref{eq:N_agent_transition_heter} indicate that the reward and the transition probability of agent $i$ at time $t$ depend on both her individual information $(s_t^i, a_t^i)$ and neighborhood empirical measure $\mu^{i, W_N}_t$.

Furthermore, the policy is assumed to be stationary for simplicity and takes the Markovian form
\begin{eqnarray}\label{eq:N_agent_policy}
a^i_t \sim  \pi^i\left(\cdot|s^i_t\right)\in \mathcal{P}(\Ac), \quad i \in [N],
\end{eqnarray}
which maps agent $i$'s state to a randomized action. For each agent $i$, the space of all policies is denoted as $\Pi$.


\begin{Remark}
When $\xi_{ij}^N \equiv 1$, $i, j \in [N]$, it corresponds to classical mean-field theory with uniform interactions \cite{CLT2019,GGWX2020}. Furthermore, our framework is flexible enough to include the non-uniform interactions of actions via $\nu_t^{i, W_N}= {\frac{1}{N}} \sum_{j \in \Vc_N} \xi_{i, j}^N  \delta_{a_t^j}(\cdot)$, and also to include heterogeneity of agents by allowing $r$ and $P$ to rely on the agent types $i$.
\end{Remark}

The objective of the multi-agent system \reff{equ:densegraph}-\reff{eq:N_agent_policy} is to maximize the expected discounted accumulated reward averaged over all agents, i.e.,
\beq\label{eq:N_agent_reward}
V_N(\mu)&=&\sup_{(\pi^1, \ldots, \pi^N) \in \Pi^N}J_N(\mu, \pi_1, \ldots, \pi_N) \\
&:=& \sup_{(\pi^1, \ldots, \pi^N) \in \Pi^N}\frac{1}{N}\sum_{i=1}^N\mathbb{E} \left.\left[\sum_{t=0}^\infty \gamma^t {r\big(s^i_t, \,\,\mu^{i,W_N}_t, \,\,a^i_t\big)}\,\,\right|\,\,s_0^i \sim \mu, a_t^i \sim \pi^i(\cdot|s_t^i)\right], \nonumber
\enq
subject to \reff{equ:densegraph}-\eqref{eq:N_agent_policy}  with a discount factor $\gamma$ $\in$ $(0, 1)$.

\begin{Definition} An $\varepsilon$-Pareto optimality of cooperative MARL \reff{equ:densegraph}-\reff{eq:N_agent_reward} for $\varepsilon >0$ is defined as $(\pi^{1, *}, \ldots, \pi^{N, *}) \in \Pi^N$ such that
\beq
J_N(\mu, \pi_1^*, \ldots, \pi_N^*) \geq  \sup_{(\pi^1, \ldots, \pi^N) \in \Pi^N} J_N(\mu, \pi_1, \ldots, \pi_N)- \varepsilon.
\enq

\end{Definition}
\subsection{Graphon Mean-Field Control}
We expect the cooperative MARL \reff{equ:densegraph}-\reff{eq:N_agent_reward} to become a GMFC problem as $N \to \infty$.  In GMFC, there is a continuum of agents $\alpha \in \mathcal{I}$, and each agent with the index/label $\alpha \in \mathcal{I}$ follows
\beq \label{equ:GMFC}
s_0^\alpha \sim \mu^\alpha, \;\; a_t^\alpha \sim \pi^\alpha(\cdot|s_t^\alpha),\;\; s_{t + 1}^\alpha \sim P(\cdot|s_t^\alpha, \mu_t^{\alpha, W}, a_t^\alpha),
\enq
where $\mu_t^\alpha$ $= \Lc(s_t^\alpha)$, $\alpha \in \mathcal{I}$ denotes the probability distribution of $s_t^\alpha$, and $\mu_t^{\alpha, W}$ is defined as the {\it neighborhood mean-field measure} of agent $\alpha$:
\beq \label{notation:weighted_neighborhood}
\mu_t^{\alpha, W}= \int_{\mathcal{I}} W(\alpha, \beta) \mu_t^\beta d \beta \in \Bc(\Sc),
\enq
with the graphon $W$ given in Assumption \ref{assm:WN}.

To ease the sequel analysis, define the space of state distribution ensembles ${\pmb \Mc}: = \Pc(\Sc)^{\mathcal{I}}:=\{f:\mathcal{I} \to \Pc(\Sc)\}$ and the space of policy ensembles ${\pmb \Pi}: = \Pc(\Ac)^{\Sc \times \mathcal{I}}$. Then ${\pmb \mu}: = (\mu^\alpha)_{\alpha \in \mathcal{I}}$ and ${\pmb \pi}: =(\pi^\alpha)_{\alpha \in \mathcal{I}}$ are elements in ${\pmb\Mc}$ and ${\pmb \Pi}$, respectively.

The objective of GMFC is to maximize the expected discounted accumulated reward averaged over all agents $\alpha \in \mathcal{I}$
\beq\label{eq:GMFC_reward}
V({\pmb \mu}):&=&\sup_{{\pmb \pi} \in {\pmb \Pi}}J({\pmb \mu}, {\pmb \pi})\\
&=& \sup_{{\pmb \pi} \in {\pmb \Pi}}\int_{\mathcal{I}} \E\left.\left[\sum_{t=0}^\infty \gamma^t {r\big(s^\alpha_t, \,\,{\mu^{\alpha,W}_t}, \,\,a^\alpha_t\big)}\,\,\right|\,\,s_0^\alpha \sim \mu^\alpha, a_t^\alpha \sim \pi^\alpha(\cdot|s_t^\alpha)\right] d\alpha. \nonumber
\enq

\section{Main Results} \label{sec:theoretical}

\subsection{Reformulation of GMFC}
In this section, we show that GMFC \reff{equ:GMFC}-\reff{eq:GMFC_reward} can be reformulated as a MDP with deterministic dynamics and continuous state-action space ${\pmb \Mc}$ $\times$ ${\pmb \Pi}$.

\begin{Theorem} \label{thm:GMFC_reformulate} GMFC \reff{equ:GMFC}-\reff{eq:GMFC_reward} can be reformulated as
\beq\label{GMFC:valuefunction_reformulate}
V({\pmb \mu}) = \sup_{{\pmb \pi} \in {\pmb \Pi}} \sum_{t=0}^\infty\gamma^t R({\pmb \mu}_t, {\pmb \pi}),
\enq
subject to
\beq\label{GMFC:dynamics_reformulate}
\mu_{t+1}^\alpha(\cdot) = {\pmb \Phi}^\alpha({\pmb \mu}_t, {\pmb \pi})(\cdot), \; t \in \N,\; \mu_0^\alpha = \mu^\alpha, \;\alpha \in \mathcal{I},
\enq
where the aggregated reward $R: {\pmb \Mc} \times {\pmb \Pi} \to \R$ and the aggregated transition dynamics ${\pmb \Phi}: {\pmb \Mc} \times {\pmb \Pi} \to {\pmb \Mc}$ are given by
\beq \label{GMFC:aggregated_reward}
R ({\pmb \mu}, {\pmb \pi}) = \int_{\mathcal{I}} \sum_{s \in \Sc} \sum_{a \in \Ac} r(s, a, \mu^{\alpha, W}) \pi^\alpha(a|s) \mu^\alpha(s)d\alpha,\\
{\pmb \Phi}^\alpha ({\pmb \mu}, {\pmb \pi})(\cdot) = \sum_{s\in \Sc}\sum_{a \in \Ac} P(\cdot|s, \mu^{\alpha, W}, a)\pi^\alpha(a|s) \mu^\alpha(s). \label{GFC:aggregated_dynamics}
\enq
\end{Theorem}
The proof of Theorem \ref{thm:GMFC_reformulate} is similar to the proof of Lemma 2.2 in \cite{GGWX2019}. So we omit it here. 

\reff{GFC:aggregated_dynamics} and \reff{GMFC:dynamics_reformulate} indicate the evolution of the state distribution ensemble ${\pmb \mu}_t$ over time. 
That is, under the fixed policy ensemble ${\pmb \pi}$, the state distribution $\mu_{t+1}^\alpha$ of agent $\alpha$ at time $t + 1$ is fully determined by the policy ensemble ${\pmb \pi}$ and the state distribution ensemble ${\pmb \mu}_t$ at time $t$. Note that the state distribution of each agent $\alpha$ is fully coupled with state distributions of the population of all agents via the graphon $W$.

With the reformulation in Theorem \ref{thm:GMFC_reformulate}, the associated $Q$ function starting from $({\pmb \mu}, {\pmb \pi}) \in {\pmb \Mc} \times {\pmb \Pi}$ is defined as
\beq
 Q({\pmb \mu}, {\pmb \pi}) &=& R({\pmb \mu}, {\pmb \pi}) + \sup_{{\pmb \pi}' \in {\pmb\Pi}} \Big[\sum_{t=1}^\infty \gamma^t R\big({\pmb \mu}_t, {\pmb \pi}'\big)\,\, \Big|\,\, s^{\alpha}_0 \sim {\mu}^{\alpha},  a_0^{\alpha} \sim \pi^{\alpha}(\cdot|s_0^{\alpha})\Big].
\enq
Hence its Bellman equation is given by
\beq\label{GMFC:Q_function}
Q({\pmb \mu}, {\pmb \pi}) = R({\pmb \mu}, {\pmb \pi}) + \gamma \sup_{{\pmb \pi}' \in {\pmb \Pi}} Q({\pmb \Phi}(\pmb \mu, \pmb \pi), {\pmb \pi}').
\enq

\begin{Remark}{(Label-state formulation)}  GMFC \reff{equ:GMFC}-\reff{eq:GMFC_reward} can be viewed as a classical MFC with extended state space $\Sc \times \mathcal{I}$, action space $\A$, policy $\tilde\pi \in \Pc(\A)^{\Sc \times \mathcal{I}}$, mean-field information $\tilde\mu \in \Pc(\Sc \times \mathcal{I})$, reward $\tilde r((s, \alpha),  \tilde\mu, a): = r(s, \int_{\mathcal{I}} W(\alpha, \beta)\tilde\mu(\cdot, \beta)d\beta, a)$, transition dynamics of $(\tilde s_t, \alpha_t)$ such that
\beqs
\tilde s_{t + 1} \sim P(\cdot|\tilde s_t, \tilde a_t, \int_{\mathcal{I}}W(\alpha_t, \beta)\tilde\mu_t(\cdot, \beta)d\beta),\; \alpha_{t+1}=\alpha_t, \; \tilde a_t \sim \tilde\pi(\cdot|\tilde s_t, \alpha_t), \; \tilde s_0 \sim \mu_0, \;\tilde \alpha_0 \sim Unif(0, 1).
\enqs
It is worth pointing out such a label-state formulation has also been studied in GMFG \cite{LS2022, CK2021}.
\end{Remark}

\subsection{Approximation}
In this section, we show that GMFC \reff{equ:GMFC}-\reff{eq:GMFC_reward} provides a good approximation for the cooperative multi-agent system \reff{equ:densegraph}-\reff{eq:N_agent_reward} in terms of the value function and the optimal policy ensemble. To do this, the following assumptions on $W$, $P$, $r$, and ${\pmb \pi}$ are needed.

\begin{Assumption}[graphon $W$] \label{assm:W} There exists $L_W >0$ such that for all $\alpha, \alpha', \beta, \beta' \in \mathcal{I}$
\begin{eqnarray*}
|W(\alpha, \beta) - W(\alpha', \beta')| \leq L_W \cdot \Big(|\alpha - \alpha'| + |\beta -\beta'|\Big).
\end{eqnarray*}
\end{Assumption}
Assumption \ref{assm:W} is common in graphon mean-field theory \cite{GC2019a, CK2021, LS2022}. Indeed, the Lipschitz continuity assumption on $W$ in Assumption \ref{assm:W} can be relaxed to piecewise Lipschitz continuity on $W$.

\begin{Assumption}[transition probability $P$] \label{assm:P} There exists $L_P >0$ such that for all $s\in\Sc, a\in\Ac$, $\mu_1,\mu_2\in\Bc(\Sc)$
\begin{eqnarray*}
\|P(\cdot| s, \mu_1, a) - P(\cdot| s, \mu_2, a)\|_1 \leq L_P \cdot \|\mu_1 - \mu_2\|_1,
\end{eqnarray*}
where $\|\cdot\|_1$ denotes $L^1$ norm here and throughout the paper.

\end{Assumption}

\begin{Assumption}[reward $r$] \label{assm:r}
 There exist $M_r>0$ and $L_{r}>0$ such that for all $s\in\Sc, a\in\Ac$, $\mu_1,\mu_2\in\Bc(\Sc)$,
    \begin{eqnarray*}
    |{r}(s,\mu, a)|\leq M_r, \; |{r}(s,\mu_1, a)- r(s,\mu_2,a)|\leq L_{r}\cdot||\mu_1-\mu_2||_1.
    \end{eqnarray*}
\end{Assumption}

\begin{Assumption}[policy ${\pmb \pi}$] \label{assm:pi} There exists $L_{\pmb \Pi} >0$ such that for any policy ensemble ${\pmb \pi}: = (\pi^\alpha)_{\alpha \in \mathcal{I}} \in {\pmb \Pi}$ is Lipschitz continuous, i.e.
\beqs
\max_{s \in \Sc}\|\pi^\alpha(\cdot|s) - \pi^\beta(\cdot|s)\|_1 \leq L_{\pmb \Pi}|\alpha -\beta|.
\enqs
\end{Assumption}

Assumptions \ref{assm:P}-\ref{assm:pi} are standard and commonly used to bridge the multi-agent system and mean-field theory.

To show approximation properties of GMFC in the large-scale multi-agent system, we need to relate policy ensembles of GMFC to policies of the multi-agent system. On one hand, one can see that any ${\pmb \pi} \in {\pmb \Pi}$ leads to a $N$-agent policy tuple $(\pi^1, \ldots, \pi^N) \in \Pi^N$ with
\beq \label{equ_relation_pmbpi_piN} 
\Gamma^N: {\pmb \Pi} \ni {\pmb \pi} \mapsto (\pi^1, \ldots, \pi^N) \in \Pi^N, \;\;\; \mbox{with}\; \pi^i: = {\pmb \pi}^{\frac{i}{N}}.
\enq
On the other hand, any $N$-agent policy tuple $(\pi^1, \ldots, \pi^N) \in \Pi^N$ can be seen as a step policy ensemble ${\pmb \pi}^N$ in ${\pmb \Pi}$:
\beq \label{equ:relation_piN_pmbpi}
{\pmb \pi}^{N, \alpha} := \sum_{i =1}^N \pi^i {\bf 1}_{\alpha \in (\frac{i-1}{N}, \frac{i}{N}]} \in {\pmb \Pi}.
\enq
\begin{Theorem}[Approximate Pareto Property] \label{thm:GMFC_approximate_pareto_property}
Assume Assumptions \ref{assm:WN}, \ref{assm:W}, \ref{assm:P}, \ref{assm:r} and \ref{assm:pi}. Then under either the condition \reff{C1} or \reff{C2}, we have for any initial distribution $\mu \in \Pc(\Sc)$
\beq
|V_N(\mu) - V({\mu})| \to 0, \;\; \text{as}\; N \to \infty.
\enq
Moreover, if the graphon convergence in Assumption \ref{assm:WN} is at rate $\mathcal{O}(\frac{1}{\sqrt{N}})$, then $|V_N(\mu) - V({\mu})| = \mathcal{O}(\frac{1}{\sqrt{N}})$.
As a consequence, for any $\varepsilon >0$, there exists an integer $N_\varepsilon$ such that when $N \geq N_\varepsilon$, the optimal policy ensemble of GMFC denoted as ${\pmb \pi^*}$ (if it exists) provides an $\varepsilon$-Pareto optimality $(\pi^{1, *}, \ldots, \pi^{N, *}): = \Gamma^N({\pmb \pi}^*)$ for the multi-agent system \reff{eq:N_agent_reward}, with $\Gamma^N$ defined in \reff{equ_relation_pmbpi_piN}.
\end{Theorem}
Directly learning Q function of GMFC in \reff{GMFC:Q_function} will lead to high complexity. Instead, we will introduce a smaller class of GMFC with a lower dimension in the next section, which enables a scalable algorithm.

\subsection{Algorithm Design} \label{sec:algorithm}
This section will show that discretizing the graphon index $\alpha \in \mathcal{I}$ of GMFC enables to approximate Q function in \reff{GMFC:Q_function} by an approximated Q function in \reff{discretized GMFC: Q_function} below defined on a smaller space, which is critical for designing efficient learning algorithms.

Precisely,  we choose uniform grids $\alpha_m \in \mathcal{I}_M:=\{\frac{m}{M}, 0 \leq m \leq M\}$ for simplicity, and approximate each agent $\alpha \in \mathcal{I}$ by the nearest $\alpha_m \in \mathcal{I}_M$ close to it. Introduce $\widetilde{\pmb \Mc}_M:= \Pc(\Sc)^{\mathcal{I}_M}$, $\widetilde{\pmb \Pi}_M: =\Pc(\Ac)^{\Sc \times \Ic_M}$. Meanwhile, $\tilde{\pmb \mu}:=(\tilde\mu^{\alpha_m})_{m \in [M]} \in \widetilde{\pmb \Mc}_M$ and $\tilde{\pmb \pi}:= (\tilde\pi^{\alpha_m})_{m \in [M]} \in \widetilde{\pmb \Pi}_M$ can be viewed as a piecewise constant state distribution ensemble in ${\pmb \Mc}$ and a piecewise constant policy ensemble in ${\pmb \Pi}$, respectively.  Our arguments can be easily generalized to nonuniform grids.

Consequently, instead of performing algorithms according to \reff{GMFC:Q_function} with a continuum of graphon labels directly, we work with GMFC with  $M$ blocks called {\bf block GMFC}, in which agents in the same block are homogeneous. The Bellman equation for Q function of block GMFC is given by
\beq \label{discretized GMFC: Q_function}
\widetilde Q(\tilde {\pmb \mu}, \tilde{\pmb \pi}) = \widetilde R(\tilde{\pmb \mu}, \tilde{\pmb \pi}) + \gamma \sup_{\widetilde {\pmb \pi}' \in \widetilde{\pmb \Pi}_M} \widetilde Q(\widetilde{\pmb \Phi}(\tilde{\pmb \mu}, \tilde{\pmb \pi}), \tilde{\pmb \pi}'),
\enq
where the neighborhood mean-field measure, the aggregated reward $\widetilde R: \widetilde{\pmb \Mc}_M \times  \widetilde{\pmb \Pi}_M \to \R$ and the aggregated transition dynamics $\widetilde {\pmb \Phi}: \widetilde{\pmb \Mc}_M \times  \widetilde{\pmb \Pi}_M \to \widetilde{\pmb \Mc}_M$ are given by
\beq \label{notation:discretized_GMFC}
\tilde \mu^{\alpha_{m}, W}&=&\frac{1}{M}\sum_{m'=0}^{M-1} W(\alpha_{m}, \alpha_{m'})\tilde\mu^{\alpha_{m'}}, m \in [M],\\
\widetilde R(\tilde{\pmb \mu}, \tilde{\pmb \pi})&=&\frac{1}{M}\sum_{m=0}^{M-1}\sum_{s \in \Sc} \sum_{a \in \Ac} r(s, a, \tilde\mu^{\alpha_m, W}) \tilde\mu^{\alpha_m}(s) \tilde\pi^{\alpha_m}(a|s),\\
\widetilde {\pmb \Phi}^{\alpha_m}(\tilde{\pmb \mu}, \tilde{\pmb \pi})(\cdot)&=&\sum_{s \in \Sc} \sum_{a \in \Ac} P(\cdot|s, a, \tilde\mu^{\alpha_m, W}) \tilde\mu^{\alpha_m}(s) \tilde\pi^{\alpha_m}(a|s). \label{notation:discretized_GMFC_Phi}
\enq

We see from \reff{discretized GMFC: Q_function} that block GMFC is a MDP with deterministic dynamics $\widetilde{\pmb \Phi}$ and continuous state-action space $\widetilde{\pmb \Mc}_M \times  \widetilde{\pmb \Pi}_M$.
The following Theorem shows that there exists an optimal policy ensemble of block GMFC in $\widetilde {\pmb \Pi}_M$.
\begin{Theorem}[Existence of Optimal Policy Ensemble] \label{thm:GMFC_existence_pareto_optimality} Given Assumptions \ref{assm:P}, \ref{assm:r}, assume $\gamma \cdot (L_P +1) < \infty$, then for any fixed integer $M >0$, there exists an $\tilde {\pmb \pi}^*$ $\in$ $\widetilde {\pmb \Pi}_M$ that maximize $\widetilde Q(\tilde{\pmb \mu}, \tilde {\pmb \pi})$ in \reff{discretized GMFC: Q_function} for any $\tilde {\pmb \mu} \in \widetilde {\pmb \Mc}_M$.
\end{Theorem}

Furthermore, we show that with sufficiently fine partitions of the graphon index $\mathcal{I}$, i.e., $M$ is sufficiently large, block GMFC \reff{discretized GMFC: Q_function}-\reff{notation:discretized_GMFC_Phi} well approximates the multi-agent system in Section \ref{subsec:local_setting}.
\begin{Theorem} \label{thm:discretized_GMFC_approximate_pareto_property}
Assume $\gamma \cdot (L_P + 1) < 1$ and Assumptions \ref{assm:WN}, \ref{assm:W}, \ref{assm:P}, \ref{assm:r} and \ref{assm:pi}. Under either \reff{C1} or \reff{C2}, for any $\varepsilon>0$, there exists $N_\varepsilon$, $M_\varepsilon$ such that for $N \geq N_\varepsilon$, the optimal policy ensemble $\tilde{\pmb\pi}^*$ of block GMFC \reff{discretized GMFC: Q_function} with $M_\varepsilon$ blocks provides an $\varepsilon$-Pareto optimality  $(\tilde\pi^{1, *}, \ldots, \tilde\pi^{N, *}): = \Gamma^N(\tilde{\pmb \pi}^*)$ for the multi-agent system \reff{eq:N_agent_reward} with $N$ agents.
\end{Theorem}

 Theorem \ref{thm:discretized_GMFC_approximate_pareto_property} shows that the optimal policy ensemble of block GMFC is near-optimal for {\it all} sufficiently large multi-agent systems, meaning that block GMFC provides a good approximation for the multi-agent system.


Recall that block GMFC can be viewed as a MDP with deterministic dynamics and continuous state-action space. To learn block GMFC,  one can adopt a similar kernel-based Q learning method in \cite{GGWX2020} for MFC, a uniform discretization method or deep reinforcement algorithms like DDPG \cite{LHPal2016} for MFC in \cite{CLT2019} with theoretical guarantees. Since block GMFC has a higher dimension than classical MFC, we choose to adapt DRL algorithm Proximal Policy Optimization (PPO) \cite{schulman2017proximal} to block GMFC and then apply the learned policy ensemble of block GMFC to the multi-agent system to validate our theoretical findings. We describe the deployment of block GMFC in the multi-agent system in Algorithm \ref{alg:1}, which we call it \textbf{N-agent GMFC}. 



\begin{algorithm}[H]
\caption{N-agent GMFC}\label{alg:1}
\begin{algorithmic}
\State \textbf{Input} Initial state distribution $\mu_0$, number of agents $N$, episode length $T$, the learned policy $\tilde{\pmb \pi} \in \widetilde{\pmb \Pi}_M$ learned by PPO
\State \textbf{Initialize} $s^i_0 \sim \mu_0$, $i \in [N]$
\For {$t$ $=$ $1$ to $T$}
  \For {$i$ $=$ $1$ to $N$}
\State Choose $m(i) = \argmin_{m \in [M]}|\frac{i}{N} - \frac{m}{M}|$
\State Sample action $a^i_t \sim {\tilde\pi}^{\alpha_{m(i)}}(\cdot|s_t^i)$, observe reward $r_t^i$ and new state $s_{t+1}^i$
  \EndFor
\EndFor
\end{algorithmic}
\end{algorithm}

\section{Proofs of Main Results} \label{sec:proofs}
In this section, we will provide proofs of Theorems \ref{thm:GMFC_approximate_pareto_property}-\ref{thm:discretized_GMFC_approximate_pareto_property}.
\subsection{Proof of Theorem \ref{thm:GMFC_approximate_pareto_property}}
To prove Theorem \ref{thm:GMFC_approximate_pareto_property}, we need the following two Lemmas. We start by defining the step state distribution ${\pmb \mu}_t^N: = ({\mu}_t^{N, \alpha})_{\alpha \in \mathcal{I}}$ for notational simplicity
\beq\label{equ:step_state_distribution}
 {\mu}_t^{N, \alpha}(\cdot) = \sum_{i \in \Vc_N} \delta_{s_t^i}(\cdot) {\bf 1}_{\alpha \in (\frac{i-1}{N}, \frac{i}{N}]}.
\enq

Lemma \ref{lemma:appendix_mu_error} shows the convergence of the neighborhood empirical measure to the neighborhood mean-field measure.

\begin{Lemma} \label{lemma:appendix_mu_error} Assume Assumptions \ref{assm:WN}, \ref{assm:W}, \ref{assm:P} and \ref{assm:pi}. Under either condition \reff{C1} or \reff{C2}, for any policy ensemble ${\pmb \pi} \in {\pmb \Pi}$, we have
\beq \label{equ:appendix_mu_error}
\sum_{i=1}^N \int_{(\frac{i-1}{N}, \frac{i}{N}]}\E\big[\|\mu_t^{i, W_N} - \mu_t^{\alpha, W}\|_1\big] d\alpha \to 0, \;\;\; \mbox{as}\; N \to \infty,
\enq
where $\mu_t^i = \mu_t^\alpha \equiv \mu \in \Pc(\Sc)$.
\end{Lemma}
Moreover, if the graphon convergence in Assumption \ref{assm:WN} is at rate $\mathcal{O}(\frac{1}{\sqrt{N}})$, then $$\sum_{i=1}^N \int_{(\frac{i-1}{N}, \frac{i}{N}]}\E\big[\|\mu_t^{i, W_N} - \mu_t^{\alpha, W}\|_1\big]d\alpha = \mathcal{O}(\frac{1}{\sqrt{N}}).$$
\begin{proof}[Proof of Lemma \ref{lemma:appendix_mu_error}] 
We first prove  \reff{equ:appendix_mu_error} under the condition \reff{C1} and then show  \reff{equ:appendix_mu_error} also holds under the condition \reff{C2}.\\
 {\bf Case 1: \; $\xi_{i,j}^N= W_N(\frac{i}{N}, \frac{j}{N})$}.\; Note that under the condition \reff{C1}, $\mu_t^{i, W_N} = \int_{\mathcal{I}} W_N(\frac{i}{N}, \beta) {\mu}_t^{N, \beta} d\beta$ by the definition of $\mu_t^{N, \alpha}$ in \reff{equ:step_state_distribution}.  Then 
\beqs
& & \sum_{i=1}^N \int_{(\frac{i-1}{N}, \frac{i}{N}]}\E\big[\|\mu_t^{i, W_N} - \mu_t^{\alpha, W}\|_1\big] d\alpha\\
&= &  \sum_{i=1}^N \int_{(\frac{i-1}{N}, \frac{i}{N}]}\E\Big[\Big\|\int_{\mathcal{I}} W_N(\frac{i}{N}, \beta) {\mu}_t^{N, \beta} d\beta -\int_{\mathcal{I}} W(\alpha, \beta){\mu}_t^{\beta}d\beta \Big\|_1\Big]d\alpha\\
&\leq  &  \sum_{i=1}^N \int_{(\frac{i-1}{N}, \frac{i}{N}]}\E\Big[\Big\|\int_{\mathcal{I}} W_N(\frac{i}{N}, \beta) {\mu}_t^{N, \beta} d\beta -\int_{\mathcal{I}} W_N(\frac{i}{N}, \beta){\mu}_t^{\beta}d\beta \Big\|_1\Big]d\alpha\\
& & \;\;\; + \; \sum_{i=1}^N \int_{(\frac{i-1}{N}, \frac{i}{N}]}\E\Big[\Big\|\int_{\mathcal{I}} W_N(\frac{i}{N}, \beta) {\mu}_t^{\beta} d\beta -\int_{\mathcal{I}} W(\alpha, \beta){\mu}_t^{\beta}d\beta \Big\|_1\Big]d\alpha\\
&= & : I_1 + I_2.
\enqs
For the term $I_1$, we adopt Theorem 2 in \cite{CK2021} and have that under the policy ensemble ${\pmb \pi}$ and $N$-agent policy $(\pi^1, \ldots, \pi^N): =\Gamma_N({\pmb \pi})$, with $\Gamma_N$ defined in \reff{equ_relation_pmbpi_piN} 
\beqs
 I_1 = \E\Big[\Big\|\int_{\mathcal{I}} W_N(\frac{i}{N}, \beta) {\mu}_t^{N, \beta} d\beta -\int_{\mathcal{I}} W_N(\frac{i}{N}, \beta){\mu}_t^{\beta}d\beta \Big\|_1\Big] \to 0, \; \mbox{as}\; N \to \infty.
\enqs
Moreover, if the graphon convergence in Assumption \ref{assm:WN} is at rate $\mathcal{O}(\frac{1}{\sqrt{N}})$, then the term $I_1$ is also at rate $\mathcal{O}(\frac{1}{\sqrt{N}})$.\\
By noting that $W_N(\alpha, \beta) = W_N\big(\frac{\lceil N\alpha \rceil}{N}, \frac{\lceil N\beta \rceil}{N}\big)$,
\beqs
I_2 &=& \sum_{i=1}^N \int_{(\frac{i-1}{N}, \frac{i}{N}]}\Big\|\int_{\mathcal{I}} W_N\big(\frac{\lceil N\alpha \rceil}{N}, \beta\big) {\mu}_t^{\beta} d\beta -\int_{\mathcal{I}} W(\alpha, \beta){\mu}_t^{\beta}
d\beta \Big\|_1 d\alpha\\
&=& \sum_{i=1}^N \int_{(\frac{i-1}{N}, \frac{i}{N}]}\Big\|\int_{\mathcal{I}} W_N\big(\alpha, \beta\big) {\mu}_t^{\beta} d\beta -\int_{\mathcal{I}} W(\alpha, \beta){\mu}_t^{\beta}d\beta \Big\|_1 d\alpha\\
&=& \int_{\mathcal{I}} \Big\|\int_{\mathcal{I}} W_N\big(\alpha, \beta\big) {\mu}_t^{\beta} d\beta -\int_{\mathcal{I}} W(\alpha, \beta){\mu}_t^{\beta}d\beta \Big\|_1 d\alpha\\
&=& \sum_{s \in \Sc} \int_{\mathcal{I}}\Big|\int_{\mathcal{I}} W_N\big(\alpha, \beta\big) {\mu}_t^{\beta}(s) d\beta -\int_{\mathcal{I}} W(\alpha, \beta){\mu}_t^{\beta}(s)d\beta \Big|d\alpha\\
&\to & 0,
\enqs
where the last inequality is from the fact in \cite{L2012} that the convergence of $\|W_N - W\|_\square \to 0$ is equivalent to the convergence of
$$\|W_N- W\|_{L_\infty \to L_1} := \sup_{\|g\|_\infty \leq 1} \int_{\mathcal{I}}\biggl|\int_{\mathcal{I}}\big(W_N(\alpha, \beta) - W(\alpha, \beta) \big)g(\beta)d\beta\biggl|d\alpha \to 0.$$
Combining $I_1$ and $I_2$, we prove \reff{equ:appendix_mu_error} under the condition \reff{C1}.

\medskip

\noindent {\bf Case 2: $\xi_{i,j}^N$ are random variables with Bernoulli($W_N(\frac{i}{N}, \frac{j}{N})$)}. 
\beqs
 & & \sum_{i=1}^N \int_{(\frac{i-1}{N}, \frac{i}{N}]} \mathbb{E}\|\mu^{i,W_N}_t-\mu^{\alpha, W}_t\|_{1}d\alpha \\
&=&  \sum_{i=1}^N \int_{(\frac{i-1}{N}, \frac{i}{N}]} \mathbb{E} \big\|\frac{1}{N}\sum_{j=1}^N \xi_{ij}^N \delta_{s_t^j} - \int_{\mathcal{I}} W(\alpha, \beta) \mu_t^{\beta} d\beta\big\|_1 d\alpha\\
&\leq &  \sum_{i=1}^N \int_{(\frac{i-1}{N}, \frac{i}{N}]} \mathbb{E} \big\|\frac{1}{N}\sum_{j=1}^N \xi_{ij}^N \delta_{s_t^j} - \int_{\mathcal{I}} W_N(\frac{i}{N}, \beta) \mu_t^{N, \beta} d\beta\big\|_1 d\alpha\\
& & \;\;\; + \; \sum_{i=1}^N \int_{(\frac{i-1}{N}, \frac{i}{N}]} \mathbb{E} \big\|\int_{\mathcal{I}} W_N(\frac{i}{N}, \beta) \mu_t^{N, \beta} d\beta - \int_{\mathcal{I}} W(\alpha, \beta) \mu_t^{\beta} d\beta\big\|_1 d\alpha\\
& =:& I_1 + I_2.
\enqs
Note from {\bf Case 1} that $I_2 \to 0$ as $N \to \infty$ and $I_2 = \mathcal{O}(\frac{1}{\sqrt{N}})$ if the graphon convergence in Assumption \ref{assm:WN} is at rate $\mathcal{O}(\frac{1}{\sqrt{N}})$. Therefore, it is enough to estimate $I_1$.
\beqs
I_1 &=&\mathbb{E} \big\|\frac{1}{N}\sum_{j=1}^N \xi_{ij}^N \delta_{s_t^j} - \int_{\mathcal{I}} W_N(\frac{i}{N}, \beta) \mu_t^{N, \beta} d\beta\big\|_1 \\
& \leq &  \mathbb{E} \Big[ \mathbb{E}\Big[\sup_{f: \Sc \to \{-1, 1\}} \Big\{\frac{1}{N} \sum_{j=1}^N \xi_{ij}^N f(s_t^j) - \frac{1}{N}\sum_{j=1}^N W_N(\frac{i}{N}, \frac{j}{N}) f(s_t^j)\Big\}\Big|s_t^1, \ldots, s_t^N\Big]\Big].
\enqs
We proceed the same argument as in the proof of Theorem 6.3 in \cite{GGWX2020}. Precisely, conditioned on $s_t^1, \ldots, s_t^N$, $\Big\{\xi_{ij}^N f(s_t^j) - W_N(\frac{i}{N}, \frac{j}{N}) f(s_t^j)\Big\}_{j=1}^N$ is a sequence of independent mean-zero random variables bounded in $[-1, 1]$ due to $\E[\xi_{i,j}^N]= W_N(\frac{i}{N}, \frac{j}{N})$. This implies that each $\xi_{ij}^N f(s_t^j) - W_N(\frac{i}{N}, \frac{j}{N}) f(s_t^j)$ is a sub-Gaussian with variance bounded by 4. As a result, conditioned on $s_t^1, \ldots, s_t^N$, $\Big\{\frac{1}{N} \sum_{j=1}^N \xi_{ij}^N f(s_t^j) - \frac{1}{N}\sum_{j=1}^N W_N(\frac{i}{N}, \frac{j}{N}) f(s_t^j)\Big\}_{i=1}^N$ is a mean-zero sub-Gaussian random variable with variance $\frac{4}{N}$. By the equation (2.66) in \cite{W2019}, we have
\beqs
I_1 &\leq& \mathbb{E} \Big[ \mathbb{E}\Big[\sup_{f: \Sc \to \{-1, 1\}} \Big\{\frac{1}{N} \sum_{j=1}^N \xi_{ij}^N f(s_t^j) - \frac{1}{N}\sum_{j=1}^N W_N(\frac{i}{N}, \frac{j}{N}) f(s_t^j)\Big\}\Big|s_t^1, \ldots, s_t^N\Big]\Big]\\
&\leq & \frac{\sqrt{8\ln(2) |\Sc|}}{\sqrt{N}}.
\enqs
Therefore, combining $I_1$ and $I_2$ in {\bf Case 2}, we show that when $\xi_{i,j}^N$ are random variables with Bernoulli($W_N(\frac{i}{N}, \frac{j}{N})$), \reff{equ:appendix_mu_error} holds under the condition \reff{C2}.
\end{proof}

\medskip

Lemma \ref{lemma:appendix_state_error} shows the convergence of the state distribution of $N$-agent game to the state distribution of GMFC.

\begin{Lemma} \label{lemma:appendix_state_error}
Assume Assumptions \ref{assm:WN}, \ref{assm:W}, \ref{assm:P} and \ref{assm:pi}. For any uniformly bounded family $\Gc$ of functions $g: \Sc \to \R$, we have
\beq \label{equ:appendix_state_error}
\sup_{g \in \Gc}\sum_{i=1}^N \int_{(\frac{i-1}{N}, \frac{i}{N}]}\Big|\E[g(s_t^i) - g(s_t^\alpha)]\Big| \to 0,
\enq
where $s_0^i \sim \mu_0$, $s_0^\alpha \sim \mu_0$. Moreover, if the graphon convergence in Assumption \ref{assm:WN} is at rate $\mathcal{O}(\frac{1}{\sqrt{N}})$, then $$\sup_{g \in \Gc}\sum_{i=1}^N \int_{(\frac{i-1}{N}, \frac{i}{N}]}\Big|\E[g(s_t^i) - g(s_t^\alpha)]\Big|= \mathcal{O}(\frac{1}{\sqrt{N}}).$$
\end{Lemma}
\begin{proof}[Proof of Lemma \ref{lemma:appendix_state_error}] The proof is by induction as follows. To do this, first introduce
\beqs
l_g(s, \mu, \pi):= \sum_{a \in \A} \sum_{s' \in \Sc} g(s') P(s'|s, \mu, a) \pi(a|s).
\enqs
\reff{equ:appendix_state_error} holds obviously at $t =0$. Suppose that \reff{equ:appendix_state_error} holds at $t$. Then for any uniformly bounded function $g$ with $|g| \leq M_g$ at $t+1$
\beq
& & \sum_{i=1}^N \int_{(\frac{i-1}{N}, \frac{i}{N}]}\Big|\E[g(s_{t+1}^i) - g(s_{t+1}^\alpha)]\Big|d\alpha \nonumber\\
&= & \sum_{i=1}^N \int_{(\frac{i-1}{N}, \frac{i}{N}]}\Big|\E\big[l_g(s_t^i, \mu_t^{i, W_N}, \pi^i)\big] - \E\big[l_g(s_t^\alpha, \mu_t^{\alpha, W}, \pi^\alpha)\big]\Big|d\alpha \nonumber \\
&\leq & \sum_{i=1}^N \int_{(\frac{i-1}{N}, \frac{i}{N}]}\Big|\E\big[l_g(s_t^i, \mu_t^{i, W_N}, \pi^i)\big] - \E\big[l_g(s_t^i, \mu_t^{\alpha, W}, \pi^i)\big]\Big|d\alpha \nonumber\\
& & \;\;\;\; +\;   \sum_{i=1}^N \int_{(\frac{i-1}{N}, \frac{i}{N}]}\Big|\E\big[l_g(s_t^i, \mu_t^{\alpha, W}, \pi^i)\big] - \E\big[l_g(s_t^\alpha, \mu_t^{\alpha, W}, \pi^i)\big]\Big|d\alpha \nonumber \\
& & \;\;\;\;+\;   \sum_{i=1}^N \int_{(\frac{i-1}{N}, \frac{i}{N}]}\Big|\E\big[l_g(s_t^\alpha, \mu_t^{\alpha, W}, \pi^i)\big] - \E\big[l_g(s_t^\alpha, \mu_t^{\alpha, W}, \pi^\alpha)\big]\Big|d\alpha \nonumber\\
&=&: I + II + III, \label{estimate _appendix_A2}
\enq
where the first equality is by the law of total expectation.
\paragraph{First term of \reff{estimate _appendix_A2}}
\beqs
 I &=&  \sum_{i=1}^N \int_{(\frac{i-1}{N}, \frac{i}{N}]}\Big|\E\big[l_g(s_t^i, \mu_t^{i, W_N}, \pi^i)\big] - \E\big[l_g(s_t^i, \mu_t^{\alpha, W}, \pi^i)\big]\Big|d\alpha\\
 &\leq& M_g L_P  \sum_{i=1}^N \int_{(\frac{i-1}{N}, \frac{i}{N}]}\E\big[\|\mu_t^{i, W_N} - \mu_t^{\alpha, W}\|_1\big]d\alpha\\
 &\to & 0, \;\;\; \mbox{as}\; N \to \infty
\enqs
where the second inequality is from the continuity of $P$, and the last inequality is from Lemma \ref{lemma:appendix_mu_error}.
\paragraph{Second term of \reff{estimate _appendix_A2}} One can view $l_g(s, \mu_t^{\alpha, W}, \pi^i)$ as a function of $s \in \Sc$ for any fixed $\mu_t^{\alpha, W}$ and $\pi^i$. Note that $|l_g(s, \mu_t^{\alpha, W}, \pi^i)|\leq M_g$, where $M_g$ is a constant independent of $\mu_t^{\alpha, W}$ and $\pi^i$. Since \reff{equ:appendix_state_error} holds at $t$, then 
\beqs
II &=&  \sum_{i=1}^N \int_{(\frac{i-1}{N}, \frac{i}{N}]}\Big|\E\big[l_g(s_t^i, \mu_t^{\alpha, W}, \pi^i)\big] - \E\big[l_g(s_t^\alpha, \mu_t^{\alpha, W}, \pi^i)\big]\Big|d\alpha\\
&\to & 0, \;\;\; \mbox{as}\; N \to \infty.
\enqs
\paragraph{Third term of \reff{estimate _appendix_A2}}
\beqs
III &=&  \sum_{i=1}^N \int_{(\frac{i-1}{N}, \frac{i}{N}]}\Big|\E\big[l_g(s_t^\alpha, \mu_t^{\alpha, W}, \pi^i)\big] - \E\big[l_g(s_t^\alpha, \mu_t^{\alpha, W}, \pi^\alpha)\big]\Big|d\alpha\\
&\leq & M_g  \sum_{i=1}^N \int_{(\frac{i-1}{N}, \frac{i}{N}]}\E\big[\|\pi^i(s_t^\alpha) - \pi^\alpha(s_t^\alpha)\|_1\big]d\alpha\\
&\leq & M_g L_{\Pi}  \sum_{i=1}^N \int_{(\frac{i-1}{N}, \frac{i}{N}]}\max_{\alpha \in (\frac{i-1}{N}, \frac{i}{N}]}|\frac{i}{N}-\alpha|d\alpha\\
&=& \mathcal{O}(\frac{1}{N}),
\enqs
where the second inequality is by the uniform boundedness of $g$ and the third inequality is from Assumption \ref{assm:pi}.
\end{proof}

Now we are ready to prove Theorem  \ref{thm:GMFC_approximate_pareto_property}.  We start by defining $\widehat r$ the aggregated reward over all possible actions under the policy $\pi$
\beqs
\widehat r(s, \mu, \pi): = \sum_{a \in \A} r(s, \mu, a) \pi(a|s).
\enqs

\begin{proof}[Proof of Theorem  \ref{thm:GMFC_approximate_pareto_property}]
\begin{eqnarray}
   & &  |V_N(\mu) - V(\mu)| \nonumber\\
   & & = \biggl|\sup_{\Pi^N}\frac{1}{N}\sum_{i=1}^N\mathbb{E} \left[\sum_{t=0}^\infty \gamma^t {r\big(s^i_t, \,\,{\mu^{i,W_N}_t}, \,\,a^i_t\big)}\right] - \sup_{{\pmb \pi \in {\pmb \Pi}}}\int_{\mathcal{I}} \E\left[\sum_{t=0}^\infty \gamma^t {r\big(s^\alpha_t, \,\,{\mu^{\alpha, W}_t}, \,\,a^\alpha_t\big)}\right] d\alpha\biggl| \nonumber\\
   & & \leq \sup_{{\pmb \pi \in \pmb \Pi}}\biggl|\frac{1}{N}\sum_{i=1}^N\mathbb{E} \left[\sum_{t=0}^\infty \gamma^t {r\big(s^i_t, \,\,{\mu^{i,W_N}_t}, \,\,a^i_t\big)}\right] - \int_{\mathcal{I}} \E\left[\sum_{t=0}^\infty \gamma^t {r\big(s^\alpha_t, \,\,{\mu^{\alpha, W}_t}, \,\,a^\alpha_t\big)}\right] d\alpha \biggl| \nonumber\\
   & & = \sup_{{\pmb \pi \in \pmb \Pi}} \biggl|\sum_{t=0}^\infty \gamma^t \sum_{i=1}^N \int_{(\frac{i-1}{N}, \frac{i}{N}]} \Big(\mathbb{E}\big[\widehat r(s^i_t, \,\,{\mu^{i,W_N}_t}, \,\,\pi^i)\big] -\mathbb{E}\big[\widehat r(s^\alpha_t, \,\,{\mu^{\alpha,W}_t}, \,\, \pi^\alpha) \big]\Big)d\alpha \biggl| \nonumber\\
       & & \leq  \sup_{{\pmb \pi \in \pmb \Pi}} \biggl|\sum_{t=0}^\infty \gamma^t \sum_{i=1}^N \int_{(\frac{i-1}{N}, \frac{i}{N}]} \Big(\mathbb{E}\big[\widehat r(s^i_t, \,\,{\mu^{i,W_N}_t}, \,\,\pi^i)\big] -\mathbb{E}\big[\widehat r(s^i_t, \,\,{\mu^{\alpha,W}_t}, \,\, \pi^i) \big]\Big)d\alpha \biggl| \nonumber \\
   & & \;\;\; \; +\; \sup_{{\pmb \pi \in \pmb \Pi}} \biggl|\sum_{t=0}^\infty \gamma^t \sum_{i=1}^N \int_{(\frac{i-1}{N}, \frac{i}{N}]} \Big(\mathbb{E}\big[\widehat r(s^i_t, \,\,{\mu^{\alpha, W}_t}, \,\,\pi^i)\big] -\mathbb{E}\big[\widehat r(s^\alpha_t, \,\,{\mu^{\alpha, W}_t}, \,\, \pi^i) \big]\Big)d\alpha \biggl|  \nonumber\\
   & & \;\;\;\; + \; \sup_{{\pmb \pi \in \pmb \Pi}} \biggl|\sum_{t=0}^\infty \gamma^t \sum_{i=1}^N \int_{(\frac{i-1}{N}, \frac{i}{N}]} \Big(\mathbb{E}\big[\widehat r(s^\alpha_t, \,\,{\mu^{\alpha, W}_t}, \,\, \pi^i) \big] - \mathbb{E}\big[\widehat r(s^\alpha_t, \,\,{\mu^{\alpha, W}_t}, \,\, \pi^\alpha) \big]\Big)d\alpha \biggl|  \nonumber\\
   & & := I + II + III, \label{proof:estimate_error}
\end{eqnarray}
where we use \reff{equ:relation_piN_pmbpi} in the second inequality.
\paragraph{First term of \reff{proof:estimate_error}}
\beq
I &\leq& \sup_{{\pmb \pi}} L_{r} \sum_{t=0}^\infty \gamma^t \sum_{i=1}^N \int_{(\frac{i-1}{N}, \frac{i}{N}]} \mathbb{E}\|\mu^{i,W_N}_t-\mu^{\alpha, W}_t\|_{1}d\alpha \nonumber\\
&= & \mathcal{O}(\frac{1}{\sqrt{N}}), \label{proof:estimate_error_firstterm}
\enq
where the last equality is from Lemma \ref{lemma:appendix_mu_error} when the convergence in Assumption \ref{assm:WN} is at rate $\mathcal{O}(1/\sqrt{N})$.

\paragraph{Second term of \reff{proof:estimate_error}} From Lemma \ref{lemma:appendix_state_error}, we have $II = \mathcal{O}(\frac{1}{\sqrt{N}})$.

\paragraph{Third term of \reff{proof:estimate_error}}
\beqs
III &\leq& \sup_{{\pmb \pi}} \sum_{t=0}^\infty \gamma^t \sum_{i=1}^N \int_{(\frac{i-1}{N}, \frac{i}{N}]} \max_{s \in \Sc} \|\pi^i(s) - \pi^\alpha(s)\|_1 d\alpha\\
&\leq & L_{\Pi} \sup_{{\pmb \pi}} \sum_{t=0}^\infty \gamma^t \sum_{i=1}^N \int_{(\frac{i-1}{N}, \frac{i}{N}]} |\frac{i}{N} -\pi^\alpha| d\alpha\\
&=& \mathcal{O}(\frac{1}{N}).
\enqs
Therefore, combining $I$, $II$ and $III$ yields the desired result.
\end{proof}

\subsection{Proof of Theorem \ref{thm:GMFC_existence_pareto_optimality}}
First, we see that \reff{discretized GMFC: Q_function} corresponds to the following optimal control problem
\beq
& & \widetilde V^M(\tilde {\pmb \mu})
:=\sup_{\tilde{\pmb \pi} \in \widetilde {\pmb \Pi}_M}\tilde J^M(\tilde{\pmb \mu}, \tilde{\pmb \pi}) \nonumber\\
& & = \sup_{\tilde{\pmb \pi} \in \widetilde {\pmb \Pi}_M}\frac{1}{M} \sum_{m=1}^M \E\left.\left[\sum_{t=0}^\infty \gamma^t {r\big(\tilde s^{\alpha_m}_t, \, {\tilde \mu^{\alpha_m,W}_t}, \,\tilde a^{\alpha_m}_t\big)}\,\right|\tilde s^{\alpha_m}_0 \sim {\tilde \mu^{\alpha_m}},\, \tilde a_t^{\alpha_m} \sim \tilde \pi^{\alpha_m}(\cdot|\tilde s_t^{\alpha_m})\right]. \label{eq:discretized_GMFC_reward}
\enq

The associated Q function of \reff{eq:discretized_GMFC_reward} is defined as
\beq \label{deftildeQ}
\tilde Q(\tilde{\pmb \mu}, \tilde{\pmb \pi}) &=& \sup_{\tilde{\pmb \pi}'}\frac{1}{M} \sum_{m=1}^M \E\left.\left[\sum_{t=0}^\infty \gamma^t {r\big(\tilde s^{\alpha_m}_t, \,{\tilde \mu^{\alpha_m,W}_t}, \,\tilde a^{\alpha_m}_t\big)}\,\right|\, \tilde s^{\alpha_m}_0 \sim \tilde{\mu}^{\alpha_m}, \tilde a_0^{\alpha_m} \sim \tilde \pi^{\alpha_m}(\cdot|\tilde s_t^{\alpha_m})\right] \nonumber\\
&=&  R(\tilde{\pmb \mu}, \tilde{\pmb \pi})+ \sup_{\tilde{\pmb \pi}' \in \widetilde{\pmb \Pi}_M}\sum_{t=1}^\infty \gamma^t \tilde R(\tilde{\pmb \mu}_t, \tilde{\pmb \pi}'),
\enq
subject to $\tilde{\pmb \mu}_{t+1} = \widetilde{\pmb \Phi}(\tilde{\pmb \mu}_{t}, \tilde{\pmb \pi})$, $\tilde{\pmb \mu}_0 = \tilde{\pmb \mu}$.

\medskip

We first show the verification result and then prove the continuity property of $\tilde Q$ in \reff{deftildeQ}, which thus leads to Theorem \ref{thm:GMFC_existence_pareto_optimality}.

\begin{Lemma}[Verification]\label{lemma:verification} Assume Assumption \ref{assm:r}. Then $\tilde Q$ in \reff{deftildeQ} is the unique function satisfying the Bellman equation \reff{discretized GMFC: Q_function}. Furthermore, if there exists $\tilde{\pmb \pi}^* \in \arg\max_{\widetilde{\pmb \Pi}_M} \tilde Q(\tilde{\pmb\mu}, \tilde{\pmb \pi})$ for each $\tilde{\pmb\mu} \in \widetilde{\pmb \Mc}_M$, then $\tilde{\pmb \pi}^*$  is an optimal stationary policy ensemble.
\end{Lemma}

The proof of Lemma \ref{lemma:verification} is standard and very similar to the proof of Proposition 3.3 in \cite{GGWX2020}.

\begin{proof}[Proof of Lemma \ref{lemma:verification}]First, define $\frac{M_r}{1-\gamma}$-bounded function space $\mathcal{Q}:=\{f: \widetilde{\pmb \Mc}_M \times \widetilde{\pmb \Pi}_M \to [-\frac{M_r}{1-\gamma}, \frac{M_r}{1-\gamma}]\}$. Then we define a Bellman operator $B: \mathcal{Q} \to \mathcal{Q}$
\beqs
(B q)(\tilde{\pmb \mu}, \tilde{\pmb\pi}): = \widetilde R(\tilde{\pmb \mu}, \tilde{\pmb\pi}) + \gamma \sup_{\tilde{\pmb\pi}' \in \widetilde{\pmb \Pi}_M}q(\widetilde {\pmb \Phi}(\tilde{\pmb \mu}, \tilde{\pmb\pi}), \tilde{\pmb\pi}'),
\enqs
One can show that $B$ is a contraction operator with the module-$\gamma$. By Banach fixed point theorem, $B$ admits a unique fixed point. As $\tilde Q$ function of \reff{deftildeQ} satisfies $B\tilde Q= \tilde Q$, $\tilde Q$ is unique solution of \reff{discretized GMFC: Q_function}.

We next define $B^{\tilde{\pmb \pi}'}: \mathcal{Q} \to \mathcal{Q}$ under the policy ensemble $\tilde {\pmb \pi}' \in \widetilde{\pmb \Pi}_M$ with
\beqs
(B^{\tilde{\pmb\pi}'} q)(\tilde{\pmb \mu}, \tilde{\pmb\pi}): = \widetilde R(\tilde{\pmb \mu}, \tilde{\pmb\pi}) + \gamma q(\widetilde{\pmb \Phi}(\tilde{\pmb \mu}, \tilde{\pmb\pi}), \tilde{\pmb\pi}').
\enqs
Similarly, we can show that $B^{\tilde{\pmb\pi}'}$ is a contraction map with the module-$\gamma$ and thus admits a unique fixed point, which is denoted as $\tilde Q^{\tilde{\pmb\pi}'}$. From this, we have
\beqs
\tilde Q^{\tilde{\pmb\pi}^*}(\tilde{\pmb\mu},  \tilde{\pmb\pi}) &=& \widetilde R(\tilde{\pmb \mu}, \tilde{\pmb\pi}) + \gamma \tilde Q^{\tilde{\pmb\pi}^*}(\tilde{\pmb \Phi}(\tilde{\pmb \mu}, \tilde{\pmb\pi}), \tilde{\pmb\pi}^*)\\
&=& \widetilde R(\tilde{\pmb \mu}, \tilde{\pmb\pi}) + \gamma \sup_{ \tilde{\pmb\pi}' \in \widetilde{\pmb \Pi}_M} \tilde Q(\widetilde{\pmb \Phi}(\tilde{\pmb \mu}, \tilde{\pmb\pi}), \tilde{\pmb\pi}') = \tilde Q(\tilde{\pmb\mu},  \tilde{\pmb\pi}),
\enqs
which implies $\tilde{\pmb\pi}^*$ is an optimal policy ensemble.
\end{proof}

\begin{Lemma}\label{lemma:continuityQ} Let Assumptions \ref{assm:P}, \ref{assm:r} hold. Assume further $\gamma \cdot (1 + L_P) < 1$. Then $\tilde Q$ in \reff{deftildeQ} is continuous.
\end{Lemma}

\begin{proof}[Proof of Lemma \ref{lemma:continuityQ}]We will show that as $\tilde{\pmb \mu}_n \to \tilde{\pmb \mu}$, $\tilde{\pmb \pi}_n \to \tilde{\pmb \pi}$ in the sense that $\int_{\mathcal{I}}\|{\tilde \mu^{\alpha}} - {\tilde\mu^{\alpha}}_n\|_1d\alpha + \int_{\mathcal{I}}\max_{s \in \Sc}\|{\tilde \pi^{\alpha}} - {\tilde\pi^{\alpha}}_n\|_1d\alpha \to 0$,
\beqs
\tilde Q(\tilde{\pmb \mu}_n, \tilde{\pmb \pi}_n) \to \tilde Q(\tilde{\pmb \mu}, \tilde{\pmb \pi}).
\enqs
From \reff{deftildeQ},
\beqs
& & |\tilde Q(\tilde{\pmb \mu}_n, \tilde{\pmb \pi}_n) -\tilde Q(\tilde{\pmb \mu}, \tilde{\pmb \pi})|\\
&\leq & \Big|\widetilde R(\tilde{\pmb \mu}, \tilde{\pmb \pi})+ \sup_{\tilde{\pmb \pi}' \in \widetilde {\pmb \Pi}_M}\sum_{t=1}^\infty \gamma^t \tilde R(\tilde{\pmb \mu}_t, \tilde{\pmb \pi}') - \widetilde R(\tilde{\pmb \mu}_n, \tilde{\pmb \pi}_n)+ \sup_{\tilde{\pmb \pi}' \in \widetilde {\pmb \Pi}_M}\sum_{t=1}^\infty \gamma^t \tilde R(\tilde{\pmb \mu}_{n, t}, \tilde{\pmb \pi}')\Big|\\
&\leq & \Big|\widetilde R(\tilde{\pmb \mu}, \tilde{\pmb \pi})- \widetilde R(\tilde{\pmb \mu}_n, \tilde{\pmb \pi}_n)\Big| + \sup_{\tilde{\pmb \pi}' \in \widetilde{\pmb \Pi}_M}\sum_{t=1}^\infty \gamma^t\Big|\tilde R(\tilde{\pmb \mu}_{n, t}, \tilde{\pmb \pi}') - \tilde R(\tilde{\pmb \mu}_t, \tilde{\pmb \pi}')\Big|\\
&\leq & L_r \cdot \int_{\mathcal{I}}\|{\tilde \mu^{\alpha, W}} - {\tilde\mu^{\alpha, W}}_n\|_1d\alpha + M_r \cdot \int_{\mathcal{I}}\|{\tilde \mu^{\alpha}} - {\tilde\mu^{\alpha}}_n\|_1d\alpha + M_r \cdot \int_{\mathcal{I}}\max_{s \in \Sc}\|{\tilde \pi^{\alpha}} - {\tilde\pi^{\alpha}}_n\|_1d\alpha \\
& &   + \sup_{\tilde{\pmb \pi}' \in \widetilde{\pmb \Pi}_M}\sum_{t=1}^\infty \gamma^t \cdot \Big(L_r \cdot \int_{\mathcal{I}}\|{\tilde \mu^{\alpha, W}}_t - {\tilde\mu^{\alpha, W}}_{n, t}\|_1d\alpha + M_r \cdot \int_{\mathcal{I}}\|{\tilde \mu^{\alpha}}_t - {\tilde\mu^{\alpha}}_{n, t}\|_1d\alpha\Big)\\
&\leq & \big(L_r + M_r\big) \cdot \int_{\mathcal{I}}\|{\tilde \mu^{\alpha}} - {\tilde\mu^{\alpha}}_n\|_1d\alpha + M_r \cdot \int_{\mathcal{I}}\max_{s \in \Sc}\|{\tilde \pi^{\alpha}} - {\tilde\pi^{\alpha}}_n\|_1d\alpha\\
& & + \sup_{\tilde{\pmb \pi}' \in \widetilde{\pmb \Pi}_M}\sum_{t=1}^\infty \gamma^t \cdot \big(L_r + M_r\big) \cdot \int_{\mathcal{I}}\|{\tilde \mu^{\alpha}}_t - {\tilde\mu^{\alpha}}_{n, t}\|_1d\alpha.
\enqs
By induction, we obtain
\beqs
\int_{\mathcal{I}}\|{\tilde \mu^{\alpha}}_t - {\tilde\mu^{\alpha}}_{n, t}\|_1d\alpha \leq (L_P + 1)\cdot \int_{\mathcal{I}}\|{\tilde \mu^{\alpha}}_{t-1} - {\tilde\mu^{\alpha}}_{n, {t-1}}\|_1d\alpha \leq \ldots \leq (L_P + 1)^{(t-1)} \int_{\mathcal{I}}\|{\tilde \mu^{\alpha}}_1 - {\tilde\mu^{\alpha}}_{n, 1}\|_1d\alpha.
\enqs
Therefore, if $\gamma\cdot(1 + L_P) < 1$, then
\beqs
 |\tilde Q(\tilde{\pmb \mu}_n, \tilde{\pmb \pi}_n) -\tilde Q(\tilde{\pmb \mu}, \tilde{\pmb \pi})| \leq C \Big(\int_{\mathcal{I}}\|{\tilde \mu^{\alpha}} - {\tilde\mu^{\alpha}}_n\|_1d\alpha + \int_{\mathcal{I}}\max_{s \in \Sc}\|{\tilde \pi^{\alpha}} - {\tilde\pi^{\alpha}}_n\|_1d\alpha\Big).
\enqs
where $C$ is a constant depending on $L_r, M_r, L_P$.
\end{proof}

Now we prove Theorem \ref{thm:GMFC_existence_pareto_optimality}.
\begin{proof}[Proof of Theorem \ref{thm:GMFC_existence_pareto_optimality}]
By Lemma \ref{lemma:continuityQ}, along with the compactness of $\widetilde{\pmb \Pi}_M$, there exists $\tilde{\pmb \pi}^* \in \widetilde{\pmb \Pi}_M$ such that $\tilde{\pmb \pi}^*\in \argmax_{\tilde{\pmb\pi} \in \widetilde{\pmb \Pi}_M} Q(\tilde{\pmb \mu}, \tilde{\pmb\pi})$. By Lemma \ref{lemma:verification}, there exists an optimal policy ensemble $\tilde{\pmb \pi}^* \in \widetilde{\pmb \Pi}_M$.
\end{proof}

\subsection{Proof of Theorem \ref{thm:discretized_GMFC_approximate_pareto_property}}
We first prove the following Lemma, which shows that GMFC and block GMFC become increasingly close to each other as the number of blocks becomes larger.
\begin{Lemma}\label{lemma:appendix_M_approximation} Under Assumptions \ref{assm:W}, \ref{assm:P} and \ref{assm:pi}, we have
\beqs
& &\sum_{m=1}^M\int_{(\frac{m-1}{M}, \frac{m}{M}]}\|\mu^{\alpha, W}_t - \tilde\mu^{\alpha_m, W}_t\|_1 d\alpha \leq \Big[(1 + L_P)^t - 1\Big] \frac{L_\Pi + 2L_PL_W + L_W}{M} + \frac{2L_W}{M},\\
& & \sum_{m=1}^M\int_{(\frac{m-1}{M}, \frac{m}{M}]}\|\mu^{\alpha}_t - \tilde\mu^{\alpha_m}_t\|_1 d\alpha \leq \Big[(1 + L_P)^t - 1\Big] \frac{L_\Pi + 2L_PL_W + L_W}{M}.
\enqs
\end{Lemma}

\begin{proof}[Proof of Lemma \ref{lemma:appendix_M_approximation}]
\beq \label{blockGMFC_GMFC_graphonmu}
& & \sum_{m=1}^M\int_{(\frac{m-1}{M}, \frac{m}{M}]}\|\mu^{\alpha, W}_t - \tilde\mu^{\alpha_m, W}_t\|_1 d\alpha\\
& \leq & \sum_{m=1}^M\int_{(\frac{m-1}{M}, \frac{m}{M}]}\|\mu^{\alpha, W}_t - \mu^{\alpha_m, W}_t\|_1 d\alpha + \frac{1}{M}\sum_{m=1}^M\|\mu^{\alpha_m, W}_t - \bar\mu^{\alpha_m, W}_t\|_1   \nonumber \\
& & \;\;+\; \frac{1}{M}\sum_{m=1}^M \|\bar\mu^{\alpha_m, W}_t - \tilde\mu^{\alpha_m, W}_t\|_1,  \nonumber 
\enq 
where $\bar\mu^{\alpha_m, W}: = \frac{1}{M} \sum_{m'=1}^M W(\alpha_m, \alpha_{m'}) \mu^{\alpha_{m'}}$.\\
By the definition of $\mu_t^{\alpha, W}, \mu_t^{\alpha_m, W}$ in \reff{notation:weighted_neighborhood}, $\tilde\mu_t^{\alpha_m, W}$ in \reff{notation:discretized_GMFC} and $\bar\mu^{\alpha_m, W}$, together with the Lipschitz continuity of $W$ in Assumption \ref{assm:W},
\beqs 
\sum_{m=1}^M\int_{(\frac{m-1}{M}, \frac{m}{M}]}\|\mu^{\alpha, W}_t - \mu^{\alpha_m, W}_t\|_1 d\alpha  &\leq& \sum_{m=1}^M\int_{(\frac{m-1}{M}, \frac{m}{M}]}\|\mu^{\alpha}_t - \mu^{\alpha_m}_t\|_1d\alpha +\frac{L_W}{M},\\
\frac{1}{M}\sum_{m=1}^M\|\mu^{\alpha_m, W}_t - \bar\mu^{\alpha_m, W}_t\|_1  &\leq& \frac{L_W}{M},\\
 \frac{1}{M}\sum_{m=1}^M \|\bar\mu^{\alpha_m, W}_t - \tilde\mu^{\alpha_m, W}_t\|_1 & \leq & \frac{1}{M}   \sum_{m=1}^M \|\mu_t^{\alpha_m} - \tilde \mu_t^{\alpha_m}\|_1.
\enqs
Plugging these into \reff{blockGMFC_GMFC_graphonmu}, 
\beqs
& & \sum_{m=1}^M\int_{(\frac{m-1}{M}, \frac{m}{M}]}\|\mu^{\alpha, W}_t - \tilde\mu^{\alpha_m, W}_t\|_1 d\alpha \leq A_t + \frac{2L_W}{M},
\enqs 
where $A_t: = \sum_{m=1}^M\int_{(\frac{m-1}{M}, \frac{m}{M}]}\|\mu^{\alpha}_t - \mu^{\alpha_m}_t\|_1 d\alpha + \frac{1}{M} \sum_{m=1}^M \|\mu_t^{\alpha_m} - \tilde \mu_t^{\alpha_m}\|_1$.

On the other hand, 
\beqs
& & \sum_{m=1}^M\int_{(\frac{m-1}{M}, \frac{m}{M}]}\|\mu^{\alpha}_t - \tilde\mu^{\alpha_m}_t\|_1d\alpha\\
& \leq & \sum_{m=1}^M\int_{(\frac{m-1}{M}, \frac{m}{M}]}\|\mu^{\alpha}_t - \mu^{\alpha_m}_t\|_1 d\alpha + \frac{1}{M} \sum_{m=1}^M \|\mu_t^{\alpha_m} - \tilde \mu_t^{\alpha_m}\|_1 = A_t.
\enqs
Therefore, it is enough to estimate $A_t$. We next estimate $A_{t + 1}$ by an inductive way. Note that $A_0=0$.
\beqs
& & A_{t +1} \\
&=& \sum_{m=1}^M\int_{(\frac{m-1}{M}, \frac{m}{M}]}\|\mu^{\alpha}_{t+1} - \mu^{\alpha_m}_{t+1}\|_1d\alpha + \frac{1}{M} \sum_{m=1}^M \|\mu_{t+1}^{\alpha_m} - \tilde \mu_{t+1}^{\alpha_m}\|_1 \\
&=& \sum_{m=1}^M \int_{(\frac{m-1}{M}, \frac{m}{M}]}\Big\|\sum_{s\in \Sc}\sum_{a \in \A} \Big(P(\cdot|s, \mu_t^{\alpha, W}, a)\pi^{\alpha}(a|s) \mu_t^{\alpha}(s)- P(\cdot|s, a, \mu^{\alpha_m, W}_t) \mu_t^{\alpha_m}(s) \pi^{\alpha_m}(a|s)\Big)\Big\|_1d\alpha\\
& & + \frac{1}{M} \sum_{m=1}^M \Big\|\sum_{s\in \Sc}\sum_{a \in \A} \Big(P(\cdot|s, \mu_t^{\alpha_m, W}, a)\pi^{\alpha_m}(a|s) \mu_t^{\alpha_m}(s)- P(\cdot|s, a, \tilde\mu^{\alpha_m, W}_t) \tilde\mu_t^{\alpha_m}(s) \tilde\pi^{\alpha_m}(a|s)\Big)\Big\|_1\\
&\leq&  \sum_{m=1}^M \int_{(\frac{m-1}{M}, \frac{m}{M}]} \Big(L_P \cdot \|\mu_t^{\alpha, W} - \mu^{\alpha_m, W}\|_1 + \frac{L_{\Pi}}{M} + \|\mu_t^{\alpha} - \mu_t^{\alpha_m}\|_1\Big) d\alpha\\
& & \; + \; \frac{1}{M} \sum_{m=1}^M \Big(L_P \cdot \|\mu_t^{\alpha_m, W} - \tilde\mu^{\alpha_m, W}\|_1 + \|\mu_t^{\alpha_m} - \tilde\mu_t^{\alpha_m}\|_1\Big)\\
&\leq & (1 + L_P) A_t + (L_\Pi + 2L_PL_W + L_W)\frac{1}{M},
\enqs
where the second equality is from \reff{GFC:aggregated_dynamics} and \reff{notation:discretized_GMFC_Phi}, and we use Assumptions \ref{assm:W}, \ref{assm:P} and \ref{assm:pi} in the third inequality.\\
By induction, we have
\beqs
A_{t + 1} \leq \Big[(1 + L_P)^t - 1\Big] \frac{L_\Pi + 2L_PL_W + L_W}{M}.
\enqs
\end{proof}

Based on Lemma \ref{lemma:appendix_M_approximation}, we have the following Proposition.
\begin{Proposition}\label{prop:discretizedGMFC_M_approximation} Assume Assumptions \ref{assm:W}, \ref{assm:P}, \ref{assm:r}, \ref{assm:pi}, and $\gamma \cdot (L_P + 1) < 1$. Then we have for any $\mu \in \Pc(\Sc)$
\beq \label{appendix:M_equivalent_class_approximate}
\sup_{{\pmb \pi} \in {\pmb \Pi}} \big|\tilde J^{M}(\mu, \pmb \pi) - J(\mu, {\pmb \pi})\big|  \to 0, \;\;\; \mbox{as}\; M \to + \infty,
\enq
where $\tilde J^{M}$ and $J$ are given in \reff{eq:discretized_GMFC_reward} and \reff{eq:GMFC_reward}, respectively. 
\end{Proposition}
\begin{proof}[Proof of Proposition \ref{prop:discretizedGMFC_M_approximation}]  Recall from \reff{notation:discretized_GMFC} that
\beqs
\widetilde J^M({\mu}, \tilde{\pmb \pi}) &=& \sum_{t=0}^\infty \gamma^t \widetilde R(\tilde{\pmb \mu}_t, \tilde{\pmb \pi}),
\enqs
subject to $\tilde \mu_{t + 1}^{\alpha_m} = \widetilde {\pmb \Phi}^{\alpha_m}(\tilde\mu_t^{\alpha_m}, \tilde\pi^{\alpha_m})$, $t \in \N_{+}$, $\tilde\mu_0^\alpha \equiv \mu$, and  $\tilde\mu^{\alpha_m,W}_t$  given in \reff{notation:discretized_GMFC}.
\beqs
J(\mu, {\pmb \pi}) &=& \sum_{t=0}^\infty \gamma^t R({\pmb \mu}_t, {\pmb \pi}),
\enqs
subject to $\mu_{t + 1}^{\alpha} = {\pmb \Phi}^{\alpha}(\mu_t^{\alpha}, \pi^{\alpha}), t \in \N_{+}$, $\mu_0^\alpha \equiv \mu$, and $\mu^{\alpha, W}_t$ given in \reff{notation:weighted_neighborhood}.
Since $\tilde{\pmb \pi}:= (\tilde\pi^{\alpha_m})_{m \in [M]} \in \widetilde{\pmb \Pi}_M$ can be viewed as a piecewise-constant projection of ${\pmb \pi} \in {\pmb \Pi}$ onto $\widetilde {\pmb \Pi}_M$.
Then,
\beqs
& & \sup_{{\pmb \pi} \in {\pmb \Pi}} \big|\tilde J^{M}(\mu, \pmb \pi) - J(\mu, {\pmb \pi})\big| \\
& \leq &  \sup_{{\pmb \pi} \in {\pmb \Pi}}  \sum_{t=0}^\infty \gamma^t \Big| \tilde R(\tilde {\pmb \mu}_t, \tilde{\pmb \pi})- R({\pmb \mu}_t, {\pmb \pi})\Big|\\
& \leq &  \sup_{{\pmb \pi} \in {\pmb \Pi}}  \sum_{t=0}^\infty \gamma^t \Big| \tilde R(\tilde {\pmb \mu}_t, \tilde {\pmb \pi})- R({\pmb \mu}_t, \tilde{\pmb \pi})\Big|+ \; \sup_{{\pmb \pi} \in {\pmb \Pi}}  \sum_{t=0}^\infty \gamma^t \Big| R({\pmb \mu}_t, \tilde {\pmb \pi})- R({\pmb \mu}_t, {\pmb \pi})\Big|\\
& := &  I + II.
\enqs
In terms of the term $I$, we first estimate $\Big| \tilde R(\tilde {\pmb \mu}_t, \tilde {\pmb \pi})- R({\pmb \mu}_t, \tilde{\pmb \pi})\Big|$:
\beqs
& & \Big| \tilde R(\tilde {\pmb \mu}_t, \tilde {\pmb \pi})- R({\pmb \mu}_t, \tilde{\pmb \pi})\Big|\\
&=& \biggl|\sum_{m=1}^M\int_{(\frac{m-1}{M}, \frac{m}{M}]}\sum_{s \in \Sc} \sum_{a \in \A} r(s, a, \tilde\mu^{\alpha_m, W}_t) \tilde\mu^{\alpha_m}_t(s) \tilde\pi^{\alpha_m}(a|s)d\alpha\\
& & \;\;\; - \sum_{m=1}^M \int_{(\frac{m-1}{M}, \frac{m}{M}]} \sum_{s \in \Sc} \sum_{a \in \A} r(s, a, \mu^{\alpha, W}_t) \mu^\alpha_t(s) \tilde\pi^{\alpha_m}(a|s)d\alpha\biggl|\\
&\leq& \biggl|\sum_{m=1}^M\int_{(\frac{m-1}{M}, \frac{m}{M}]}\sum_{s \in \Sc} \sum_{a \in \A} r(s, a, \tilde\mu^{\alpha_m, W}_t) \tilde\mu^{\alpha_m}_t(s) \tilde\pi^{\alpha_m}(a|s)d\alpha\\
& & \;\;\; - \sum_{m=1}^M \int_{(\frac{m-1}{M}, \frac{m}{M}]} \sum_{s \in \Sc} \sum_{a \in \A} r(s, a, \mu^{\alpha, W}_t) \tilde\mu^{\alpha_m}_t(s) \tilde\pi^{\alpha_m}(a|s)d\alpha\biggl|\\
& & +\; \biggl|\sum_{m=1}^M\int_{(\frac{m-1}{M}, \frac{m}{M}]}\sum_{s \in \Sc} \sum_{a \in \A} r(s, a, \mu^{\alpha, W}_t) \tilde\mu^{\alpha_m}_t(s) \tilde\pi^{\alpha_m}(a|s)d\alpha\\
& & \;\;\; - \sum_{m=1}^M \int_{(\frac{m-1}{M}, \frac{m}{M}]} \sum_{s \in \Sc} \sum_{a \in \A} r(s, a, \mu^{\alpha, W}_t) \mu^\alpha_t(s) \tilde\pi^{\alpha_m}(a|s)d\alpha\biggl|\\
& \leq & L_r \cdot \sum_{m=1}^M\int_{(\frac{m-1}{M}, \frac{m}{M}]}\|\mu^{\alpha, W}_t - \tilde\mu^{\alpha_m, W}_t\|_1 d\alpha +M_r \cdot  \sum_{m=1}^M\int_{(\frac{m-1}{M}, \frac{m}{M}]}\|\mu^{\alpha}_t - \tilde\mu^{\alpha_m}_t\|_1 d\alpha.
\enqs
By Lemma \ref{lemma:appendix_M_approximation},
\beqs
I \leq \frac{C(\gamma, L_\Pi, L_P, L_W, L_r, M_r)}{M}.
\enqs
For the term $II$,
\beqs
\sup_{{\pmb \pi} \in {\pmb \Pi}}  \sum_{t=0}^\infty \gamma^t \Big| R({\pmb \mu}_t, \tilde {\pmb \pi})- R({\pmb \mu}_t, {\pmb \pi})\Big| &\leq& \sup_{{\pmb \pi} \in {\pmb \Pi}}\sum_{t=0}^\infty \gamma^t M_r \sum_{m=1}^M\int_{(\frac{m-1}{M}, \frac{m}{M}]}\max_{s \in \Sc}\|\pi^\alpha - \pi^{\alpha^m}\|_1d\alpha\\
& \leq & \frac{L_{\Pi}M_r}{1-\gamma} \frac{1}{M}.
\enqs
\end{proof}

\begin{proof}[Proof of Theorem \ref{thm:discretized_GMFC_approximate_pareto_property}]
Suppose that $\tilde{\pmb \pi}^* \in \widetilde{\pmb \Pi}_M \subset {\pmb \Pi}$ and $(\pi^{1, *}, \ldots, \pi^{N, *}) \in \Pi^N$ are optimal policies of the problems \reff{eq:discretized_GMFC_reward} and \reff{eq:N_agent_reward}, respectively. From Proposition \ref{prop:discretizedGMFC_M_approximation}, for any $\varepsilon >0$, there exists sufficiently large $M_\varepsilon>0$
\beqs
 |\tilde J^{M_\varepsilon}(\mu, \tilde {\pmb \pi}^*) - J(\mu, {\tilde {\pmb \pi}^*})| \leq \frac{\varepsilon}{3},
\enqs
where by \reff{equ:relation_piN_pmbpi}, ${\pmb \pi}^{N, *}:= \sum_{i =1}^N \pi^{i, *} {\bf 1}_{\alpha \in (\frac{i-1}{N}, \frac{i}{N}]}$.\\
From Theorem \ref{thm:GMFC_approximate_pareto_property}, for any $\varepsilon > 0$, there exists $N_\varepsilon$ such that for all $N \geq N_\varepsilon$
\beqs
|J_N(\mu, {\tilde\pi}^{1, *}, \ldots, {\tilde\pi}^{N, *}) - J(\mu, {\tilde {\pmb \pi}^*})| \leq \frac{\varepsilon}{3}, \;\; |J_N(\mu, \pi^{1, *}, \ldots, \pi^{N, *}) - J(\mu, {{\pmb \pi}^{N, *}})| \leq \frac{\varepsilon}{3}.
\enqs
Then we have
\beqs
 & & J_N(\mu, {\tilde \pi}^{1, *}, \ldots, {\tilde \pi}^{N, *})-J_N({\mu}, {\pi}^{1, *}, \ldots, {\pi}^{N, *})\\
 &\geq & \underbrace{J_N(\mu, {\tilde \pi}^{1, *}, \ldots, {\tilde \pi}^{N, *}) - J(\mu, {\tilde {\pmb \pi}^*})}_{I_1} + \underbrace{J(\mu, {\tilde {\pmb \pi}^*}) - \tilde J_{M_\varepsilon}(\mu, \tilde{\pmb\pi}^*)}_{I_2}\\
  & & \;\;\; +\; \underbrace{\tilde J^{M_\varepsilon}(\mu, \tilde{\pmb\pi}^*) - \tilde J^{M_\varepsilon}(\mu, {\pmb \pi}^{N, *})}_{I_3} + \underbrace{\tilde J^{M_\varepsilon}(\mu, {\pmb \pi}^{N, *}) - J_N({\mu}, {\pi}^{1, *}, \ldots, {\pi}^{N, *})}_{I_4}\\
 & \geq & -\frac{\varepsilon}{3}-\frac{\varepsilon}{3}-\frac{\varepsilon}{3} = -\varepsilon.
\enqs
where $I_3 \geq 0$ due to the optimality of $\tilde {\pmb \pi}^*$ for $\tilde V^{M_\varepsilon}$. This means that the optimal policy of block GMFC provides an $\varepsilon$-optimal policy for the multi-agent system with $(\tilde\pi_1^*, \ldots, \tilde\pi_N^*): = \Gamma_N(\tilde {\pmb\pi}^*)$.
\end{proof}

\section{Experiments} \label{experiment}

In this section, we provide an empirical verification of our theoretical results, with two examples adapted from existing works on learning MFGs \cite{cui2021discrete, CLK2021} and learning GMFGs \cite{CK2021}.

\subsection{SIS Graphon Model}

We consider a SIS graphon model in \cite{cui2021discrete} under a cooperative setting. In this model, each agent $\alpha \in \mathcal{I}$ shares a state space $\Sc =\{S, I\}$ and an action space $\A=\{C, NC\}$, where $S$ is susceptible, $I$ is infected, $C$ represents keeping contact with others, and $NC$ means keeping social distance. The transition probability of each agent $\alpha$ is represented as follows
\beqs
P(s_{t+1}=I | s_t = S, a_t =C, \mu_t^{\alpha, W}) &=& \beta_1 \mu_t^{\alpha, W}(I),\\
P(s_{t+1}=I | s_t = S, a_t =NC, \mu_t^{\alpha, W}) &=& \beta_2 \mu_t^{\alpha, W}(I),\\
P(s_{t+1}=S | s_t = I, \mu_t^{\alpha, W}) &=& \delta,
\enqs
where $\beta_1$ is the infection rate with keeping contact with others, $\beta_2$ is the infection rate under social distance, and $\delta$ is the fixed recovery rate. We assume $0<\beta_2 <\beta_1$, meaning that keeping social distance can reduce the risk of being infected. The individual reward function is defined as
\beqs
r(s, \mu_t^{\alpha, W}, a)=-c_1 {\bf 1}_{\{I\}}(s) - c_2{\bf 1}_{\{NC\}}(a) - c_3{\bf 1}_{\{I\}}(s){\bf 1}_{\{C\}}(a),
\enqs
where $c_1$ represents the cost of being infected such as the cost of medical treatment, $c_2$ represents the cost of keeping social distance, and $c_3$ represents the penalty of going out if the agent is infected.

In our experiment, we set $\beta_1$=0.8, $\beta_2$=0, $\delta=0.3$ for the transition dynamics and $c_1$=2, $c_2$=0.3, $c_3=0.5$ for the reward function. The initial mean field $\mu_0$ is taken as the uniform distribution. We set the episode length to 50.

\subsection{Malware Spread Graphon Model}
We consider a malware spread model in \cite{CLK2021} under a cooperative setting. In this model, let $\Sc=\{0, 1, \ldots, K -1\}$, $K \in \N$, denote the health level of the agent, where $s_t=0$ and $s_t=K-1$  represents the best level and the worst level, respectively. All agents can take two actions: $a_t = 0$ means doing nothing, and $a_t=1$ means repairing. The state transition is given by
\begin{align*}
s_{t +1} = \left\{\begin{array}{cll} s_t + \lfloor (K-s_t)\chi_t \rfloor, \; & \mbox{if}\; a_t=0,\\
0,\; &\mbox{if}\;a_t=1,
\end{array}
\right.
\end{align*}
where $\chi_t, t \in \N$ are i.i.d. random variables with a certain probability distribution. Then after taking action $a_t$, agent $\alpha$ will receive an individual reward
\beqs
r(s_t, \mu_t^{\alpha, W}, a_t)=-(c_1+\langle \mu_t^{\alpha, W} \rangle)s_t/K - c_2 a_t.
\enqs
Here considering the heterogeneity of agents, we use $W(\alpha, \beta)$ to denote the {\it importance} effect of agent $\beta$ on agent $\alpha$. $\langle \mu_t^{\alpha, W}\rangle: = \int_{\beta \in \mathcal{I}}\sum_{s \in \Sc} s W(\alpha, \beta) \mu_t^\beta(s)d\beta$ is the risk of being infected by other agents and $c_2$ is the cost of taking action $a_t$.

In our experiment, we set $K$=3, $c_1$=0.3, and $c_2$=0.5. In addition, to stabilize the training of the RL agent, we fix $\chi_t$ to a static value, i.e., 0.7. In this model, we set the episode length to 10.

\subsection{Performance of N-agent GMFC on Multi-Agent System} 
For both models, we use PPO \cite{schulman2017proximal} to train the block GMFC agent in the infinite-agent environment and obtain the policy ensembles and further use Algorithm \ref{alg:1} to deploy them in the finite-agent environment. We test the performance of N-agent GMFC with 10 blocks to different numbers of agents, i.e., from 10 to 100. For each case, we run 1000 times of simulations and show the mean and standard variation (Green shadows in Figure \ref{figure_sis} and Figure \ref{figure_malware}) of the mean episode reward. We can see that in both scenarios and for different types of graphons,  the mean episode rewards of the N-agent GMFC become increasingly close to that of block GMFC as the number of agents grows. (See Figure \ref{figure_sis} and Figure \ref{figure_malware}). This verifies our theoretical findings empirically.

\begin{figure}[!htb]
    \centering

    \begin{subfigure}[b]{\textwidth}
        \centering
        \includegraphics[width=0.32\linewidth]{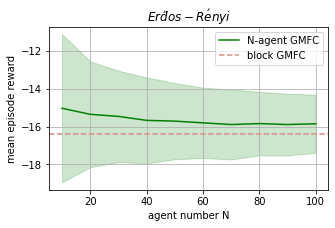}
        \hfill
        \includegraphics[width=0.32\linewidth]{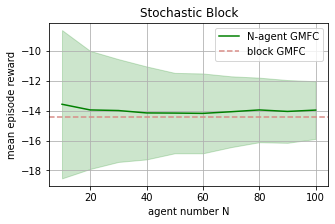}
        \hfill
        \includegraphics[width=0.32\linewidth]{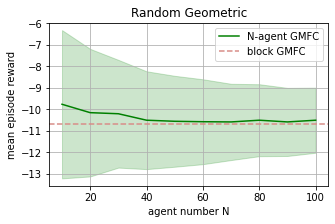}
    \end{subfigure}
    \caption{Experiments for different graphons in SIS finite-agent environment}
    \label{figure_sis}
\end{figure}

\begin{figure}[!htb]
    \centering

    \begin{subfigure}[b]{\textwidth}
        \centering
        \includegraphics[width=0.32\linewidth]{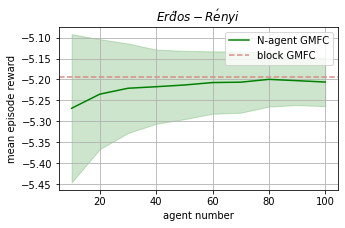}
        \hfill
        \includegraphics[width=0.32\linewidth]{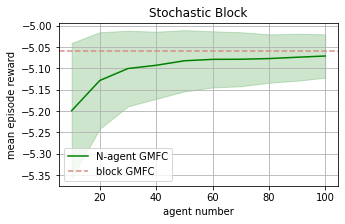}
        \hfill
        \includegraphics[width=0.32\linewidth]{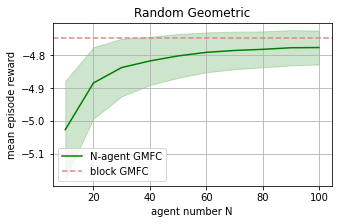}
    \end{subfigure}
    \caption{Experiments for different graphons in Malware Spread finite-agent environment} \label{figure_malware}
\end{figure}

\subsection{Comparison with Other Algorithms}
For different types of graphons, we compare our algorithm N-agent GMFC with three existing MARL algorithms, including two independent learning algorithms, i.e., independent DQN \cite{mnih2013playing}, independent PPO \cite{schulman2017proximal} and a powerful centralized-training-and-decentralized-execution(CTDE)-based algorithm QMIX \cite{rashid2018qmix}. We test the performance of those algorithms with different numbers of blocks, i.e., 2, 5, 10, to the multi-agent systems with 40 agents. The results are reported in Table \ref{table_sis} and Table \ref{table_malware}.   

In the SIS graphon model, N-agent GMFC shows dominating performance in most cases and outperforms independent algorithms by a large margin. Only QMIX can reach comparable results. And in the malware spread graphon model, N-agent GMFC outperforms other algorithms in more than half of the cases. Only independent DQN has comparable performance in this environment. And we can see that in both environments, the performance gap between N-agent GMFC and other MARL algorithms is shrinking as the number of blocks goes larger. This is mainly because the action space of block GMFC increases more quickly than MARL algorithms as the block number increases. And it is hard to train RL agents when the action space is too large.  

Beyond the visible results shown in Tables \ref{table_sis} and \ref{table_malware}, when the number of agents $N$ grows larger, classic MARL methods become infeasible because of the curse of dimensionality and the restriction of memory storage, while N-agent GMFC is trained only once and independent of the number of agents $N$, hence is easier to scale up in a large-scale regime and enjoys a more stable performance. We can see that N-agent GMFC shows more stable results when N increases as shown in Figure \ref{figure_sis} and Figure \ref{figure_malware}.

\begin{table}[H]
  \caption{Mean Episode Reward for SIS with 40 agents}
  \label{table_sis}
  \centering
  \begin{tabular}{cccccc}
    \toprule
    \multirow{2}{*}{Graphon Type} & \multirow{2}{*}{M} &  \multicolumn{4}{c}{Algorithm}    \\
    \cmidrule(r){3-6}
     & & N-agent GMFC    &  \hspace*{2mm} I-DQN \hspace*{2mm}   & \hspace*{2mm} I-PPO \hspace*{2mm}   & \hspace*{2mm} QMIX \hspace*{2mm} \\
    \midrule
    \multirow{3}{*}{Erd\H{o}s R\'enyi}  & 2 & \textbf{-15.37}  & -17.58 & -20.63 & -20.51  \\
     & 5  & \textbf{-15.74}  & -16.17  & -20.42 & 16.94 \\
     & 10 & -15.67 & -17.55 & -21.38  & \textbf{-14.45}   \\
    \multirow{3}{*}{Stochastic Block} & 2  & \textbf{-13.58} & -16.05 & -18.38  & -17.69  \\
     & 5  & \textbf{-13.67} & -15.91 & -20.13 & -13.79  \\
     & 10 & \textbf{-13.57} & -15.52 & -14.87 & -13.86   \\
    \multirow{3}{*}{Random Geometric } & 2  & \textbf{-12.45} & -17.93  & -14.82 & -14.52   \\
     & 5  & \textbf{-9.82} & -12.81  & -12.99 & -10.84  \\
     & 10 & \textbf{-10.52} & -11.68 & -12.66 & -12.60   \\

    \bottomrule
  \end{tabular}
\end{table}

\begin{table}[H]
  \caption{Mean Episode Reward for Malware Spread with 40 agents}
  \label{table_malware}
  \centering
  \begin{tabular}{cccccc}
    \toprule
    \multirow{2}{*}{Graphon Type} & \multirow{2}{*}{M} &  \multicolumn{4}{c}{Algorithm}    \\
    \cmidrule(r){3-6}
     & & N-agent GMFC    &  \hspace*{2mm} I-DQN \hspace*{2mm}   & \hspace*{2mm} I-PPO \hspace*{2mm}   & \hspace*{2mm} QMIX \hspace*{2mm}\\
    \midrule
    \multirow{3}{*}{Erd\H{o}s R\'enyi}  & 2 & -5.21 &\textbf{-5.11}  & -5.31 & -6.05   \\
     & 5  & \textbf{-5.21} &-5.30  & -5.26 & -6.13    \\
     & 10 & -5.21 &\textbf{-5.14} & -5.27 & -5.21     \\
    \multirow{3}{*}{Stochastic Block} & 2  & \textbf{-5.16} & -5.21  & -5.37 & -5.88   \\
     & 5  & \textbf{-5.10} & -5.19  & -5.31 & -5.70    \\
     & 10 & -5.09 &\textbf{-5.05}  & -5.28 & -5.27    \\
    \multirow{3}{*}{Random Geometric } & 2 & \textbf{-5.02} & -5.21  & -5.27 & -5.35    \\
     & 5  & \textbf{-4.85} & -5.03  & -5.04 & -5.05    \\
     & 10 & \textbf{-4.82} & -4.83  & -5.14 & -4.83    \\

    \bottomrule
  \end{tabular}
\end{table}

\subsection{Implementation Details}
 We use three graphons in our experiments: (1) Erd\H{o}s R\'enyi: $W(\alpha,\beta)=0.8$; (2) Stochastic block model: $ W(\alpha, \beta) = 0.9$, if $0 \leqslant \alpha, \beta \leqslant 0.5$ or $ 0.5 \leqslant \alpha, \beta \leqslant 1$, $W(\alpha, \beta) = 0.4$, otherwise; (3) {Random geometric graphon: $W(\alpha, \beta) = f(\min(|\beta - \alpha|, 1 - |\beta - \alpha|))$, where $f(x) = \mathrm{e}^{- \frac{x}{0.5-x} }$.}

For the RL algorithms, we use the implementation of RLlib \cite{liang2018rllib} (version 1.11.0, Apache-2.0 license). For PPO used to learn an optimal policy ensemble in block GFMC, we use a 64-dimensional  linear layer to encode the observation and 2-layer MLPs with 256 hidden units per layer for both value network and actor network. For independent DQN and independent PPO, we use the default weight-sharing model with 64-dimensional embedding layers. We train the GMFC PPO agent for 1000 iterations, and other three MARL agents for 200 iterations. The specific hyper-parameters are listed in Table \ref{alg_setting}.

\begin{table}[!htb]
		\begin{center}
			\caption{RL Algorithm Settings }
			\setlength{\tabcolsep}{3mm}{
				\begin{tabular}{l c c c c}
					\toprule  
					\textbf{Algorithms} & GMFC PPO & I-DQN & I-PPO & QMIX  \\
					\midrule  
		            Learning rate & 0.0005 & 0.0005 & 0.0001 & 0.00005 \\
					Learning rate decay & True & True & True & False \\
					Discount factor & 0.95 & 0.95 & 0.95 & 0.95\\
					Batch size & 128 & 128 & 128 & 128\\
					KL coefficient & 0.2 & - & 0.2 & - \\
					KL target & 0.01 & - & 0.01 & - \\
					Buffer size & - & 2000 & - & 2000 \\
					Target network update frequency & - & 2000 & - & 1000 \\
					\bottomrule 
				\end{tabular}
				\label{alg_setting}}
				\vspace{-.3cm}
		\end{center}
	\end{table}

\section{Conclusion}
In this work, we have proposed a discrete-time GMFC framework for MARL with nonuniform interactions on dense graphs. Theoretically, we have shown that under suitable assumptions, GMFC approximates MARL well with approximation error of order $\mathcal{O}(\frac{1}{\sqrt{N}})$. To reduce the dimension of GMFC, we have introduced block GMFC by discretizing the graphon index and shown that it also approximates MARL well. Empirical studies on several examples have verified the plausibility of the GMFC framework. For future research, we are interested in establishing theoretical guarantees of the PPO-based algorithm for block GMFC, learning the graph structure of MARL and extending our framework to MARL with nonuniform interactions on sparse graphs.

\bibliographystyle{elsarticle-harv}
\bibliography{refs}



\end{document}